%% file: main-arxiv.tex
\documentclass[a4paper,UKenglish]{lipics-report}
\usepackage{pta}

\title{Stochastic Timed Games Revisited}

\author[1]{S. Akshay\thanks{Partly supported by DST-INSPIRE Faculty Grant [IFA12-MA-17].}}
\author[2]{Patricia Bouyer\thanks{Partly supported by ERC project EQualIS (308087).}}
\author[1]{Shankara Narayanan Krishna \thanks{Partly supported by CEFIPRA project AVeRTS.}} 
\author[1]{\quad Lakshmi Manasa}
\author[3]{Ashutosh Trivedi}
\affil[1]{
  Department of Computer Science \& Engineering, IIT Bombay, India\\
  \url{{akshayss,krishnas,manasa}@cse.iitb.ac.in}} 
\affil[2]{LSV, CNRS \& ENS Cachan, Universit\'e Paris-Saclay, France\\
  \url{bouyer@lsv.fr}}
\affil[3]{University of Colorado Boulder, USA\\
  \url{ashutosh.trivedi@colorado.edu}}

\authorrunning{Akshay, Bouyer, Krishna, Manasa and Trivedi}

\newcommand\pat[1]{{\color{magenta} #1}}

\begin{document}
\maketitle

\begin{abstract}
\input{abstract.tex}
\end{abstract}

\section{Introduction}
\label{sec:intro}
\input{intro.tex}

\section{Stochastic Timed Games}
\label{sec:defn}
\label{sec:prelims}
\input{prelims.tex}

\section{Undecidability Results for Quantitative Reachability}
\label{undec}

In this section, we focus on the quantitative reachability problem for
STGs. We strengthen the existing undecidability result, which holds
for $2 \frac{1}{2}$ STGs~\cite{icalp09}, in two distinct
directions. First, we show the undecidability of the quantitative
reachability problem %(Theorem \ref{thm:undec-one})
in $1 \frac{1}{2}$ STGs, improving from $2 \frac{1}{2}$. Second, we
show the undecidability of the quantitative reachability problem for
$2 \frac{1}{2}$ STGs even in the time-bounded setting.

For both results, given a two-counter machine, we construct
respectively, $1 \frac{1}{2}$ and $2 \frac{1}{2}$ STGs whose building
blocks are the modules for the instructions in the two-counter
machine. The objective of player $\Diamond$ is linked to a faithful
simulation of various increment, decrement and zero-test instructions
of the two-counter machine by choosing appropriate delays to adjust
the clocks to reflect changes in counter values. 
 However, the two proofs
differ in how this verification is done and even in the problem from
which the reduction is done, i.e., halting/non-halting for two-counter
machines. This results in two quite different and non-trivial
reductions as described in Subsection~\ref{one} and
Subsection~\ref{two} respectively. 

\subsection{Quantitative reachability for $1 \frac{1}{2}$ STGs}
\label{one}
\input{one-and-half.tex}

\subsection{Time-bounded quantitative reachability for $2 \frac{1}{2}$ STGs}
\label{two}
\input{two-and-half-new.tex}

\section{Decidability results for quantitative reachability}
\label{sec:quant-dec}
\input{new-quant-dec.tex}

\section{Discussion}
\label{sec:discuss}
\input{discussion.tex}

\bibliographystyle{plain}
\bibliography{main}

\newpage
\appendix
\input{appendix-new.tex}

\end{document}

%% file: abstract.tex
Stochastic timed games (STGs), introduced by Bouyer and Forejt,
naturally generalize both continuous-time Markov chains and timed
automata by providing a partition of the locations between those
controlled by two players (Player Box and Player Diamond) with
competing objectives and those governed by stochastic laws. Depending
on the number of players---$2$, $1$, or $0$---subclasses of stochastic
timed games are often classified as $2\frac{1}{2}$-player,
$1\frac{1}{2}$-player, and $\frac{1}{2}$-player games where the
$\frac{1}{2}$ symbolizes the presence of the stochastic ``nature''
player. For STGs with reachability objectives it is known that
$1\frac{1}{2}$-player one-clock STGs are decidable for qualitative
objectives, and that $2\frac{1}{2}$-player three-clock STGs are
undecidable for quantitative reachability objectives. This paper
further refines the gap in this decidability spectrum. We show that
quantitative reachability objectives are already undecidable for
$1\frac{1}{2}$ player four-clock STGs, and even under the time-bounded
restriction for $2\frac{1}{2}$-player five-clock STGs. We also obtain
a class of $1\frac{1}{2}$, $2\frac{1}{2}$ player STGs for which the
quantitative reachability problem is decidable.

%% file: intro.tex
Two-player zero-sum games over finite state-transition graphs are a
natural framework for controller synthesis for discrete event systems.
In this setting two players---say Player Box and Player Diamond (after
necessity and possibility operators)---represent the controller and
the environment, and control-program synthesis corresponds to finding
a winning (or optimal) strategy of the controller for some given
performance objective. Finite graphs, however, often do not
satisfactorily model real-time safety-critical systems as they
disregard not only the continuous dynamics of the physical environment
but also the presence of stochastic behavior. Stochastic behavior in
such systems stems from many different sources, e.g., faulty or
unreliable sensors or actuators, uncertainty in timing delays, the
random coin flips of distributed communication and security protocols.

Timed automata~\cite{AD94} were introduced as a formalism to model
asynchronous real-time systems interacting with a continuous physical
environment. Timed automata and their two-player
counterparts~\cite{AMPS98} provide an intuitive and semantically
unambiguous way to model non-stochastic real-time systems, and a
number of case-studies~\cite{UP01} demonstrate their application in
the design and analysis of real-time systems. On the other hand,
classical formalisms (discrete-time and continuous-time) Markov
decision processes (MDPs) and stochastic games~\cite{Put94,FV97}
naturally model analysis and synthesis problems for stochastic
systems, and have been applied in control theory, operations research,
and economics.

For the formal analysis of stochastic real-time systems, a number of
recent works considered a combination of stochastic features with
timed automata, e.g. probabilistic timed automata~\cite{KNSS02},
continuous probabilistic timed automata~\cite{KNSS00} and stochastic
timed automata~\cite{journal-sta}. Probabilistic timed automata,
respectively continuous probabilistic and stochastic timed automata
can be considered as generalizations of timed automata with the
features of discrete-time Markov decision processes, respectively
continuous-time Markov chains~\cite{BHHK03} (or even generalized
semi-Markov processes~\cite{BKK+11b}).  Stochastic timed
games~\cite{icalp09} form the most general formalism for studying
controller-synthesis for stochastic real-time systems. These games can
be considered as interactions between three players---Player Box,
Player Diamond and the stochastic player (Nature)---such that Player
Box and Player Diamond are adversarial and choose their delay and
action so as to maximize and minimize probability to reach a given set
of target states, while the stochastic player plays according to a
given probability distribution. A key verification problem in this
setting is that of games with reachability objectives, where the goal
of Player Diamond is to reach a set of target states, while the goal
of the Player Box is to avoid it.
 
  \noindent{\bf Related Work.} Probabilistic timed
    automata~\cite{KNSS02} and games~\cite{FKNT10} can be considered
    as subclasses of stochastic timed games where all of the locations
    controlled by stochastic players are \emph{urgent} (no time delay
    allowed), while the decision-stochastic timed automata
    of~\cite{BBG14} can be seen as a subclass of $1\frac{1}{2}$-player
    STGs where the locations of the rational players are urgent. The
    quantitative reachability problem for probabilistic timed automata
    is known to be decidable~\cite{KNSS02} with any
    number of clocks, while the best known decidability result for the
    quantitative reachability problem for $1\frac{1}{2}$-player STGs
    is using a single clock. $\frac{1}{2}$-player STGs, also called
    stochastic timed automata (STA)~\cite{journal-sta}, have also
    received considerable attention: an abstraction based on the
    region abstraction has been proposed, which allows to solve the
    qualitative reachability problem under a \emph{fairness}
    assumption on the STA (several subclasses of STAs have been proven
    to be fair).
    For quantitative reachability, the only decidability result is for
    a subclass of single-clock STA~\cite{qest08}, but a recent
    approximability result has been shown in~\cite{BBBC16} for the
    class of \emph{fair STA}.

Other variants of stochastic timed automata have been studied in
  the past. The model in~\cite{KNSS00} uses ``countdown clocks''
  (which decrease from a set value) unlike the more timed-automata
  style of clock variables used in our model. The model in~\cite{BDHK06} (which is also called stochastic timed automata; we shall refer to them here as Modest-STA) is very general and encompasses most models with time and probabilities (and in particular the STA of~\cite{journal-sta}). However, Modest-STA is more aimed at capturing general languages (and providing a tool-set to simulate their runs) and less with decidability issues, and hence is orthogonal to our approach.

\noindent{\bf Contributions.}
The scope of this paper is to investigate decidability of
  the reachability problem in
  STGs as defined in~\cite{icalp09}, for which
  the decidability picture is far from complete. In~\cite{icalp09},
  the authors showed the decidability of qualitative reachability
  problem on $1$-clock $1\frac{1}{2}$-player STGs, and the
  undecidability of quantitative reachability problem on STGs (with
  $2\frac{1}{2}$-players). This leaves a wide gap in the decidability
  horizon of STGs.  In this paper, we study $1\frac{1}{2}$,
  $2\frac{1}{2}$-player games and contribute to a better understanding
  of the decidability status of STGs with
  quantitative reachability objectives.
\begin{table}[t]
  \begin{center}
    \begin{tabular}{|c|c|c|c|}
      \hline
      Model &  & Qualitative Results & Quantitative Results \\
      \hline
      \multirow{2}{*}{$\frac{1}{2}$ player}&  1 clock & Dec.
      \cite{lics08} & Dec. (some restrictions) \cite{qest08}\\
      \cline{2-4}
      &  $n$ clocks & \begin{tabular}{@{}c@{}} Open in
          general \\ Dec. (fair)~\cite{journal-sta} \end{tabular} & \begin{tabular}{@{}c@{}} Open in
          general \\ Approx. (fair)~\cite{BBBC16} \end{tabular} \\
      \hline
      \multirow{2}{*}{$1 \frac{1}{2}$ player}&  1 clock & Dec. \cite{icalp09}& \textcolor{green!40!black}{\bf{Dec. (Initialized, Theorem \ref{thm:quant-dec}) }} \\
      \cline{2-4}
      &  $n$ clocks & Open & \begin{tabular}{@{}c@{}}
        \textcolor{green!40!black}{\bf{Undec. (Theorem
            \ref{thm:undec-one})}} \\ \textcolor{green!40!black}{Conj: \textbf{Undec.}} (Time
        bounded) \end{tabular}\\
      \hline
      \multirow{2}{*}{$2 \frac{1}{2}$ player}&  1 clock & \textcolor{green!40!black}{Conj: \textbf{Dec.}}
      & \textcolor{green!40!black}{\bf{Dec. (Initialized, Corollary~\ref{cor:dec})}} \\
      \cline{2-4}
      &  $n$ clocks & Open &
      \begin{tabular}{@{}c@{}} Undec~\cite{icalp09} \\
      \textcolor{green!40!black}{\bf{Undec. (Time bounded, Theorem
          \ref{thm:undec-two})}} \end{tabular}\\
      \hline
    \end{tabular}
\end{center}
    \caption{
     Results in bold are contributions from
      this paper.  ``Conj'' are conjectures.}
    \label{tab:res}
    \vspace{-2em}
  \end{table}

Table~\ref{tab:res} summarizes the results presented in this paper. We show that the quantitative reachability problem is already undecidable for $1\frac{1}{2}$-player games for systems with 4 or more clocks and  for $2\frac{1}{2}$-player games the quantitative reachability problem remains undecidable even under the time-bounded restriction with 5 or more clocks. Another key contribution of this paper is the characterization of a previously unexplored subclass of stochastic timed games for which we recover decidability of quantitative reachability game for $1\frac{1}{2}$ (and even $2\frac{1}{2}$)-player stochastic timed games.   We call a 1-clock stochastic timed game \emph{initialized} if (i) all the transitions from non-stochastic states to stochastic states reset the clock, and (ii) in every bounded cycle, the clock is reset. The definition can be generalized to multiple clocks using the notion of strong reset 
where one resets all the clocks together. 
For some of the  gaps in this spectrum, we provide our best conjectures  as justified in the Discussion section:--the undecidability of time-bounded quantitative reachability for $1\frac{1}{2}$-player STG, and the decidability of qualitative reachability  of 1-clock $2\frac{1}{2}$-player STG. Due to lack of space, details of some proofs can be found in 
the Appendix.

%% file: prelims.tex
We use standard notations for the set of reals ($\Real$),
  rationals ($\mathbb{Q}$), and integers ($\Int$), and add subscripts
  to indicate additional constraints (for instance $\Rplus$ is for the
  set of non-negative reals).
Let $\pclocks$ be a finite set of real-valued variables called
\emph{clocks}. A \emph{valuation} on $\pclocks$ is a function $\nu :
\pclocks \to \Rplus$.  We assume an arbitrary but fixed ordering on
the clocks and write $x_i$ for the clock with order $i$.  This allows
us to treat a valuation $\nu$ as a point $(\nu(x_1), \nu(x_2), \ldots,
\nu(x_n)) \in \Rplus^{|\pclocks|}$.  Abusing notations slightly, we
use a valuation on $\pclocks$ and a point in $\Rplus^{|\pclocks|}$
interchangeably.  For a subset of clocks $X \subseteq \pclocks$ and
valuation $\nu \in \V$, we write $\nu[X{:=}0]$ for the valuation where
$\nu[X{:=}0](x) = 0$ if $x \in X$, and $\nu[X{:=}0](x) = \nu(x)$
otherwise. For $t\in\Rplus$, write $\nu+t$ for the valuation defined
by $\nu(x)+t$ for all $x\in X$.  The valuation $\zero \in \V$ is a
special valuation such that $\zero(x) = 0$ for all $x \in \pclocks$.
A clock constraint over $\pclocks$ is a subset of
$\Rplus^{|\pclocks|}$ defined by a (finite) conjunction of constraints
of the form $x \bowtie k,$ where $k \in \Int_{\ge 0}$, $x \in
\pclocks$, and $\mathord{\bowtie} \in \{<,\leq, =, >, \geq\}$. We
write $\rect(\pclocks)$ for the set of clock constraints.
For a constraint
$g \in \rect(\pclocks)$, and a valuation $\nu$, we write $\nu \models
g$ to represent the fact that valuation $\nu$ satisfies constraint
$g$ (defined in a natural way).

A timed automaton (TA)~\cite{AD94} is a tuple $\Aa=(L, \pclocks, E,
\Ii)$ such that (i) $L$ is a finite set of locations, (ii) $\pclocks$
is a finite set of clocks, (iii) $E \subseteq L \times \rect(\pclocks)
\times 2^{\pclocks} \times L$ is a finite set of edges, (iv) $\Ii : L
\rightarrow \rect(\pclocks)$ assigns an invariant to each location. A
state $s$ of a timed automaton is a pair $s=(\ell, \nu) \in L \times
\Rplus^{|\pclocks|}$ such that $\nu \models {\Ii}(\ell)$ (the clock
valuation should satisfy the invariant of the location).  If $s=(\ell,
\nu)$, and $t \in \Rplus$, we write $s+t$ for the state $(\ell,
\nu+t)$.  A transition $(t,e)$ from a state $s=(\ell, \nu)$ to a state
$s'=(\ell', \nu')$ is written as $s \xrightarrow{t,e} s'$ if $e=(\ell,
g, C, \ell') \in E$, such that $\nu+t \models g$, and for every $0
\leq t' \leq t$ we have $\nu+t' \models \Ii(\ell)$ and
$\nu'=\nu+t[C{:=}0](x)$. A run is a finite or infinite sequence of
transitions $\rho=s_0 \xrightarrow{t_1,e_1}s_1 \xrightarrow{t_2,e_2}
s_2 \dots $ of states and transitions. An edge $e$ is enabled from $s$
whenever there is a state $s'$ such that $s \xrightarrow{0,e} s'$.
Given a state $s$ of $\Aa$ and an edge $e$, we define $I(s,e)=\{t \in
\Rplus \mid s \xrightarrow{t,e} s'\}$ for some $s'$ and
$I(s)=\bigcup_{e \in E}I(s,e)$.  We say that $\Aa$ is non-blocking iff
for all states $s$, $I(s) \neq \emptyset$.  Now we are ready to
introduce stochastic timed games.
\begin{definition}[Stochastic Timed Games~\cite{icalp09}]
  A \emph{stochastic timed game (STG)} is a tuple $\mathcal{G} =
  (\mathcal{A},(L_{\Box},L_{\Diamond},L_{\bigcirc}),\omega,\mu)$ where
\begin{itemize}
\item $\Aa {=} (L, \pclocks, E, \Ii)$ is a timed automaton;
\item $L_{\Box}, L_{\Diamond}$, and $L_{\bigcirc}$ form a partition of
  $L$ characterizing the set of locations controlled by players $\Box$
  and $\Diamond$ and the stochastic player, respectively;
\item $\omega: E(L_{\bigcirc}) \to \Int_{>0}$ assigns some
    positive weight to each edge originating from $L_{\bigcirc}$
    (notation $E(L_{\bigcirc})$);
\item
  $\mu$ is a function assigning a measure over $I(s)$ to all states $s$ $\in$
  $L_{\bigcirc}\times \V$ satisfying the properties that $\mu (s)(I(s))=1$ and
  for Lebesgue measure $\lambda$, if $\lambda(I(s))>0$ then 
  for each measurable set $B \subseteq I(s)$ we have $\lambda(B) = 0$ if and
  only if $\mu(s)(B) = 0$.
 \end{itemize}
\end{definition}
The timed automaton $\Aa$ is said equipped with uniform distributions
over delays if for every state $s$, $I(s)$ is bounded, and $\mu(s)$ is
the uniform distribution over $I(s)$.  The timed automaton $\Aa$ is
said equipped with exponential distributions over delays whenever, for
every state $s$, either $I(s)$ has Lebesgue measure zero, or
$I(s){=}\Rplus$ and for every location $l$, there is a positive
rational $\alpha_l$ such that $\mu(s)(I(s)){=}\int_{t \in I} \alpha_l
e^{-\alpha_l t} dt$.  
For $s \in L_{\bigcirc}\times \V$, both delays and discrete moves will
be chosen probabilistically: from $s$, a delay $t$ is chosen following
the probability distribution over delays $\mu(s)$.  Then, from state
$s+t$, an enabled edge is selected following a discrete probability
distribution that is given in a usual way with the weight function
$w$: in state $s + t$, the probability of edge $e$ (if enabled),
denoted $p(s+t)(e)$ is $w(e) / \sum_{e'} \set{w(e') \mid e'~ \mbox{is
    enabled in}~s+t}$.  This way of probabilizing behaviours in timed
automata has been presented in \cite{journal-sta}.

If $L_\Box {=} \emptyset$ then the STGs are called $1\frac{1}{2}$ STGs
or $1 \frac{1}{2}$-player STGs while STGs with $L_\Box {=} L_\Diamond
{=} \emptyset$ are called $\frac{1}{2}$ STGs or $\frac{1}{2}$-player
STGs or STAs.  We often refer to $ l {\in} L_{\bigcirc}$ as stochastic
nodes, $l \in L_{\Box}$ as box (or $\Box$) nodes and $l \in
L_{\Diamond}$ as diamond (or $\Diamond$) nodes.

Fix a STG $\mathcal{G} =
(\mathcal{A},(L_{\Box},L_{\Diamond},L_{\bigcirc}),\omega,\mu)$ with
$\Aa=(L, \pclocks, E, \Ii)$ for the rest of this section.

\smallskip
\noindent{\bf Strategies, Profiles, and Runs.}
A strategy for Player $\Box$ (resp. $\Diamond$) is a function that
maps a finite run $\rho= s_0 \xrightarrow{t_0,e_0}
s_1\xrightarrow{t_1,e_1} \dots s_n$ to a pair $(t,e)$ such that $s_n
\xrightarrow{t,e} s'$ for some state $s'$, whenever $s_n=
(\ell_n,\nu_n)$ and $\ell_n \in L_\Box$ (resp. $\ell_n \in
L_\Diamond$).  In this work we focus on deterministic strategies,
  though randomized strategies could also make sense; nevertheless
  understanding the case of deterministic strategies is already
  challenging.
A strategy profile is a pair $\Lambda=(\lambda_\Diamond,\lambda_\Box)$
where $\lambda_\Diamond, \lambda_{\Box}$ respectively are strategies
of players $\Diamond$ and $\Box$.  In order to measure probabilities
of certain sets of runs, the following measurability condition is
imposed on strategy profiles $\Lambda=(\lambda_\Diamond,
\lambda_\Box)$: for every finite sequence of edges $e_1, \dots, e_n$
and every state $s$, the function $\kappa_s : (t_1, \dots, t_n)
\rightarrow (t,e)$ defined by $\kappa_s(t_1, \dots, t_n)=(t,e)$ iff
$\Lambda(s \xrightarrow{t_1,e_1} s_1 \xrightarrow{t_2,e_2} s_2 \dots
\xrightarrow{t_n,e_n} s_n)=(t,e)$, should be measurable.

Given a finite run $\rho$ ending in state $s_0$, and a strategy
profile $\Lambda$, define $Runs(\mathcal{G}, \rho, \Lambda)$
(resp. $Runs^\omega(\mathcal{G}, \rho, \Lambda)$) to be the set of all
finite (resp. infinite) runs generated by $\Lambda$ after prefix
$\rho$; that is, the set of all runs of the automaton satisfying the
following condition: If $s_i=(\ell_i, \nu_i)$ and $\ell_i \in
L_{\Diamond}$ (resp. $\ell_i \in L_{\Box}$), then $\lambda_{\Diamond}$
(resp.  $\lambda_{\Box}$) returns $(t_{i+1}, e_{i+1})$ when applied to
$\rho \xrightarrow{t_1,e_1} s_1 \xrightarrow{t_2,e_2} \dots
\xrightarrow{t_i,e_i} s_i$.  Given a finite sequence $e_1,\dots,e_n$
of edges, a symbolic path $\pi_\Lambda(\rho,e_1\dots e_n)$ is defined
as
\[
\pi_\Lambda(\rho,e_1\dots e_n)=\{\rho' \in Runs(\mathcal{G}, \rho,
\Lambda) \mid \rho'= \rho \xrightarrow{t_1,e_1} s_1
\xrightarrow{t_2,e_2} s_2 \dots \xrightarrow{t_n,e_n} s_n,
~\mbox{with}~t_i \in \Rplus\}.
\]
When $\Lambda$ is clear, we simply write $\pi(\rho,e_1\dots e_n)$.

\smallskip
\noindent{\bf Probability Measure of a Strategy Profile.}
Given a strategy profile $\Lambda=(\lambda_{\Diamond},
\lambda_{\Box})$, and a finite run $\rho$ ending in $s=(\ell, \nu)$, a
measure $\mathcal{P}_{\Lambda}$ can be defined on the set
$Run(\mathcal{G}, \rho, \Lambda)$, following~\cite{icalp09}: First,
for the empty sequence $\epsilon$,
$\mathcal{P}_{\Lambda}(\pi(\rho,\epsilon))=1$, and
\begin{itemize} 
\item If $\ell \in L_{\Diamond}$ (resp. $\ell \in L_\Box$), and
  $\lambda_{\Diamond}(\rho) =(t,e)$ (resp. $\lambda_{\Box}(\rho)
  =(t,e)$), then $\mathcal{P}_{\Lambda}(\pi(\rho, e_1\dots e_n))$
  equals $0$ if $e_1 \neq e$ and equals
  $\mathcal{P}_{\Lambda}(\pi(\rho \xrightarrow{t,e} s', e_2\dots
  e_n))$, otherwise.
\item If $\ell {\in} L_{\bigcirc}$ then
 % \[
  $\mathcal{P}_{\Lambda}(\pi(\rho, e_1\dots e_n)) {=} \int_{t\in I(s,e_1)}
  p(s+t)(e_1)\ \cdot\ \mathcal{P}_{\Lambda}(\pi(\rho \xrightarrow{t,e_1} s',
  e_2\dots e_n))\ d\mu(s)(t)$
 % \]
  where $s \xrightarrow{t,e_1} s'$ for every $t \in I(s,e_1)$.
\end{itemize}
The cylinder generated by a symbolic path is defined as follows: an
infinite run $\rho''$ is in the cylinder generated by
$\pi_{\Lambda}(\rho, e_1, \dots, e_n)$ denoted
$\mathsf{Cyl}(\pi_{\Lambda}(\rho, e_1, \dots, e_n))$ if $\rho'' \in
Runs^{\omega}(\stg, \rho, \Lambda)$ and there is a finite prefix
$\rho'$ of $\rho''$ such that $\rho' \in \pi_{\Lambda}(\rho, e_1,
\dots, e_n)$.  It is routine to extend the above measure
$\mathcal{P}_{\Lambda}$ to cylinders, and thereafter to the generated
$\sigma$-algebra; extending~\cite{journal-sta}, one can show this is
indeed a probability measure over $Runs^{\omega}(\stg, \rho,
\Lambda)$.

 \smallskip
 \noindent{\bf Example.}
An example of a STG is shown in the adjoining
figure. In this example all the locations belong to stochastic player
(this is an $\frac{1}{2}$ STG) and there is only one clock named
$x$. \input{stg-ex} We explain here the method for computing
probabilities.  We assume uniform distribution over delays at all
states, and initial state $s_0 = (A,0)$.  Let $d\mu_{(A,0)}$ be the
uniform distribution over $[0,1]$ and $d\mu_{(B,0)}$ uniform
distribution over $[0,2]$.  Then
  $\mathcal{P}(\pi((A,0),e_1e_2))$ equals
  \[
  \int_0^1 \frac{\mathcal{P}(\pi((B,0),e_2))}{2}d\mu_{(A,0)}(t)
  =\int_0^1 \frac{1}{2}(\int_1^2\frac{1}{2}d\mu_{(B,0)}(u))\
  d\mu_{(A,0)}(t) = \frac{1}{2}\int_0^1 (\int_1^2 \frac{1}{2}
  \frac{1}{2} du))\ dt)=\frac{1}{8}.
  \]
 \smallskip
 \noindent{\bf Reachability Problem.}
 We study the reachability
problem for STGs, stated as follows.  Given a STG $\mathcal{G}$ with a
set $T$ of target locations, an initial state $s_0$ and a threshold
$\bowtie p$ with $p \in [0,1] \cap \mathbb{Q}$, decide whether there
is a strategy $\lambda_\Diamond$ for Player~$\Diamond$ such that for
every strategy $\lambda_\Box$ for Player~$\Box$,
$\mathcal{P}_{\Lambda}(\{\rho \in Run(\mathcal{G},s_0,\Lambda) \mid
\rho\ \text{visits}\ T\}) \bowtie p$, with
$\Lambda=(\lambda_\Diamond,\lambda_\Box)$.  There are two categories of
reachability questions:
\begin{enumerate}
\item \textbf{Quantitative reachability}: The constraint on
  probability involves $0 < p< 1$.
\item \textbf{Qualitative reachability}: The constraint on probability
  involves $p \in \{0,1\}$.
\end{enumerate}
The key results of the paper are the following:
\begin{theorem}
The quantitative reachability problem is 
\begin{enumerate}
\item Undecidable for $1\frac{1}{2}$ STGs with 4 or more clocks;
\item Undecidable for $2 \frac{1}{2}$ STGs with 5 or more clocks even
  under the time-bounded semantics;
\item Decidable for $1\frac{1}{2}$ and $2 \frac{1}{2}$
  \emph{initialized} STGs with one clock.
\end{enumerate}
\end{theorem}
Mentioned restrictions (time-bounded semantics and initialized)
  will be introduced when needed.
In Section~\ref{undec}, we deal with the quantitative reachability
problem, where we show strengthened undecidability results. In
Section~\ref{sec:quant-dec}, we explore a new model of STGs with a
single clock and an initialized restriction to recover decidability
for the quantitative reachability problem. In Section \ref{sec:discuss}, we 
discuss the intrinsic difficulties and challenges ahead, summarize our key contributions and conjectures.

%% file: stg-ex.tex
 \begin{wrapfigure}[5]{o}[-5pt]{.4\textwidth}
  \centering
   \scalebox{0.7}{                                                                                                                                                       
   \begin{tikzpicture}[->,thick]                                                                                                                                          
     \node[initial, cir, initial text ={},label=below: $x \leq 1$] at (0,0) (A) {A} ;                                                                                             \node[cir,label=below: $x \leq 2$] at (4,0) (B) {B}; 
     \node[cir] at (6,0) (D) {D};
     \path (A) edge node [above]{$x\leq 1,\ e_1$} node [below]{$x:\equal0$}(B);
     \path(A) edge [loop above] node [above]{$x \leq 1, e_3$}(A); 
     \path (B) edge node [above]{$x\geq 1,\ e_2$} (D);
     \path (B) edge[bend left] node [below]{$x\leq 2,\ e_4$}(A);
   \end{tikzpicture}                                                                                                                                                      
   }                                                                                                                                                                     
   \label{ex1}                                                                                                                                                        
  \end{wrapfigure}

%% file: one-and-half.tex
As mentioned above, in the case of $1 \frac{1}{2}$ STGs we improve the
corresponding result of~\cite{icalp09} for $2\frac{1}{2}$ STGs. But
unlike in~\cite{icalp09}, we reduce from the \emph{non-halting
  problem} for two-counter machines to the existence of a winning
strategy for Player $\Diamond$ with the desired objective. This
crucial difference makes it possible for the probabilistic player to
verify the simulation performed by player $\Diamond$.

\begin{theorem} \label{thm:undec-one} 
  The quantitative reachability problem is undecidable for
  $1\frac{1}{2}$ STGs with $\ge 4$ clocks.
\end{theorem}

Let ${\cal M}$ be a two-counter machine. Our reduction uses a
$1\frac{1}{2}$ player STG ${\cal G}$ with four clocks and uniform
distributions over delays, and a set of target locations $T$ such that
player $\Diamond$ has a strategy to reach $T$ with probability
$\frac{1}{2}$ iff ${\cal M}$ does not halt. Each instruction
(increment, decrement and test for zero value) is specified using a
module. The main invariant in our reduction is that upon entry into a
module, we have that $x_1=\frac{1}{2^{c_1}}, x_2=\frac{1}{2^{c_2}}$,
$x_3=x_4=0$, where $c_1$ (resp. $c_2$) is the value of counter $C_1$
(resp. $C_2$) in $\mathcal{M}$.

\begin{figure} [!h]                                                                                                                                                           \begin{center}                                                                                                                                                         
  \scalebox{0.65}{                                                                                                                                                       
    \begin{tikzpicture}[->,thick]                                                                                                                                          
      \node[initial, dia, initial text ={$x_1=\frac{1}{2^{c_1}}$}] at (-5,0) (A) {$\ell_i$} ;
      \node[dia] at (-3,0) (B) {B}; 
      \node[cir,label=above:$x_4\equal 0$] at (0,0) (C) {C};
      \node[dia] at (3,0) (D) {D};
      \node[dia] at (5,0) (F) {$\ell_j$};
      \node[dashed,rectangle,draw=red,fill=yellow] at (0,-2)(E){GetProb}; 
      
      \path (A) edge node [above]{$x_1=1$} node [below]{$\{x_1,x_4\}$} (B);
      \path (A) edge[loop below] node [below]{$x_2=1,\{x_2\}$} (A);
      \path (B) edge[loop below] node [below]{$x_2=1,\{x_2\}$} (B);
      \path (B) edge node [above]{$0{<}x_1,x_3{<}1$} node [below]{$\{x_4\}$} (C);
      \path (C) edge node[below]{$\{x_1\}$} (D);
      \path (C) edge node[left]{$\{x_2\}$} (E);
      \path (D) edge[loop above] node [above]{$x_2=1,\{x_2\}$} (D);
      \path (D) edge node [above]{$x_3=1$} node [below]{$\{x_3,x_4\}$} (F);
    \end{tikzpicture}
  }
  
  \scalebox{0.65}{
    \begin{tikzpicture}[->,thick]
      \node[cir,initial,initial where=below, initial text={}, label=left:$x_4 \leq 2$] at (0,0) (E0) {E0};
      \node[cir,accepting] at (-9,-4) (T1) {T1};
      \node[cir,accepting] at (-9, 4) (T2) {T2};
      \node[cir,accepting] at (9,-4) (T3) {T3};
      \node[cir,accepting] at (9,4) (T4) {T4};
      
      \node[cir] at (-10,-2) (R1) {R1};                  
      \node[cir] at (-10, 2) (R2) {R2};                  
      \node[cir] at (10, -2) (R3) {R3};                  
      \node[cir] at (10, 2) (R4) {R4};                  

      \node[cir,label=left:$x_4 \leq 2$] at (8,0) (P1) {P1};
      \node[cir,label=right:$x_4 \leq 2$] at (-8, 0) (P2) {P2};                                                 
      
      \node[dia] at (-6, 2) (G1) {G1};
      \node[dia] at (-6, -2) (H1) {H1};
      \node[dia] at (-4, 2) (G) {G};
      \node[dia] at (-4,-2) (H) {H};

      \node[dia] at (-2, -2) (E1) {E1};
      \node[dia] at (-2, 2) (E2) {E2};
      
      \node[dia] at (2,-2) (E3) {E3};
      \node[dia] at (2,2) (E4) {E4};

      \node[dia] at (4,-2) (I) {I};
      \node[dia] at (4,2) (J) {J};
      
      \node[dia] at (6,-2) (I1) {I1};
      \node[dia] at (6, 2) (J1) {J1};

      \path (E0) edge node [left]{$x_1 \geq 1 \wedge x_4 \leq 1$} (E1);
      \path (E0) edge node [left]{$x_3 \geq 2 \wedge x_4 \leq 2$} (E2);
      \path (E0) edge node [right]{$x_1 \leq  1$} (E3);
      \path (E0) edge node [right]{$x_4 \geq 1 \wedge x_3 \leq 2$} (E4);

      \path (E2) edge node [above]{$x_4=2$} node [below]{$\{x_2,x_4\}$} (G);
      \path (E1) edge node [above]{$x_4=2$} node [below]{$\{x_2,x_4\}$} (H);
      \path (E4) edge node [above]{$x_4=2$} node [below]{$\{x_2,x_4\}$} (J);
      \path (E3) edge node [above]{$x_4=2$} node [below]{$\{x_2,x_4\}$} (I);
      \path (G) edge[loop above] node [above]{$x_3=3,\{x_3\}$} (G);
      \path (H) edge[loop below] node [below]{$x_3=3,\{x_3\}$} (H);
      \path (J) edge[loop above] node [above]{$x_3=3,\{x_3\}$} (J);
      \path (I) edge[loop below] node [below]{$x_3=3,\{x_3\}$} (I);
      
      \path (G) edge node [above]{$x_1=3$} node [below]{$\{x_1,x_2\}$} (G1);
      \path (H) edge node [above]{$x_1=3$} node [below]{$\{x_1,x_2\}$} (H1);      
      \path (I) edge node [above]{$x_1=3$} node [below]{$\{x_1,x_2\}$} (I1);      
      \path (J) edge node [above]{$x_1=3$} node [below]{$\{x_1,x_2\}$} (J1);
      
      \path (J1) edge[bend left] node [left]{$x_4=1$} node [right, near start]{$\{x_2,x_4\}$} (P1);
      \path (I1) edge[bend right] node [left]{$x_4=1$} node [right, near start]{$\{x_2,x_4\}$} (P1);
      \path (G1) edge[bend right] node [left, near start]{$x_4=1$} node [right]{$\{x_2,x_4\}$} (P2);
      \path (H1) edge[bend left] node [left, near start]{$x_4=1$} node
            [right]{$\{x_2,x_4\}$} (P2);
            
      \path(P2) edge node[near end, right]{$x_1 \leq  1$}(T1);
      \path(P2) edge node[near end, right]{$x_4 \geq 1 \wedge x_3 \leq 2$}(T2);
      \path(P1) edge node[near end, left]{$x_1 \geq 1 \wedge x_4 \leq 1$}(T3);
      \path(P1) edge node[near end, left]{$x_3 \geq 2 \wedge x_4 \leq 2$}(T4);
      
      \path (P1) edge[bend left=10] node [right=0.2]{$x_1\leq 1$} (R3);
      \path (P1) edge node [right, near end]{$x_4\geq 1 \wedge x_3\leq 2$} (R4);
      \path (P2) edge node [left]{$x_1\geq 1\wedge x_4\leq 1$} (R1);
      \path (P2) edge node [left]{$x_3\geq 2 \wedge x_4\leq 2$} (R2);

    \end{tikzpicture}                                                                                                                                                      
  }                                                                                                                                                                     
  \caption{The Increment $c_1$ module on the top  and the GetProb
    gadget below}                                                                                                                                       
  \label{inc-c1}                                                                                                                                                        
\end{center}                                                                                                                                                         
\end{figure}
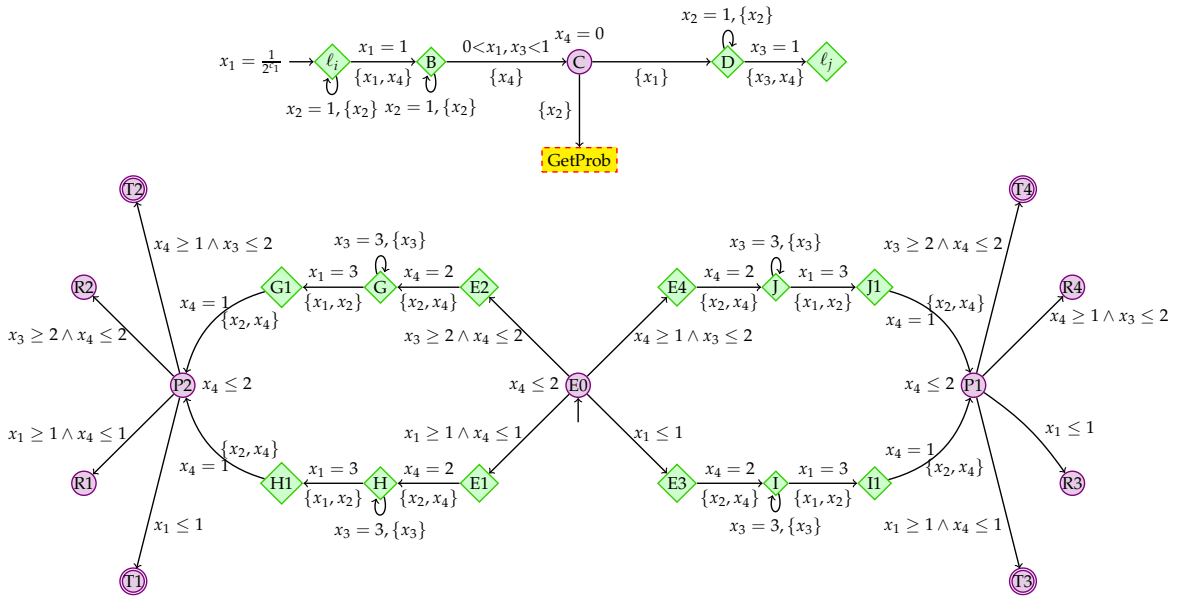 

\newcommand\Red[1]{{\textcolor{red}{#1}}}     

We outline the simulation of an increment instruction <<~$\ell_i$ :
increment counter $C_1$, goto $\ell_j$~>> in Figure \ref{inc-c1}
(top). The module is entered with values $x_1=\frac{1}{2^{c_1}},
x_2=\frac{1}{2^{c_2}}$, $x_3=x_4=0$. A time $1-\frac{1}{2^{c_1}}$ is
spent at location $\ell_i$, so that at location $B$ we have $x_1=0$,
$x_2=\frac{1}{2^{c_2}}+1-\frac{1}{2^{c_1}}$ (or
$\frac{1}{2^{c_2}}-\frac{1}{2^{c_1}}$, if $c_2>c_1$ -- we write
  in all cases $\frac{1}{2^{c_2}}+1-\frac{1}{2^{c_1}}~\text{mod}~1$),
$x_3=1-\frac{1}{2^{c_1}}$, $x_4=0$.
An amount of time $t\in (0,\frac{1}{2^{c_1}})$ is spent at $B$, which
is decided by Player~$\Diamond$. We rewrite this as
$t=\frac{1}{2^{c_1+1}}\pm \epsilon$ for $-\frac{1}{2^{c_1+1}}<
\epsilon <\frac{1}{2^{c_1+1}}$. This is because, ideally we want $t$
to be $\frac{1}{2^{c_1+1}}$ and want to consider any deviation as an
error.

Now at $C$, we have $x_1=t$,
$x_2=\frac{1}{2^{c_2}}+1-\frac{1}{2^{c_1}}+t~\text{mod}~1$,
$x_3=1-\frac{1}{2^{c_1}}+t$, $x_4=0$. The computation proceeds to $D$
with probability $\frac{1}{2}$, and the location $\ell_j$
corresponding to the next instruction $\ell_j$ is reached with $x_1=
\frac{1}{2^{c_1}}-t$, $x_2= \frac{1}{2^{c_2}}$, $x_3=x_4=0$. On the
other hand, with probability $\frac{1}{2}$, the gadget $GetProb$ is
reached. The gadget $GetProb$ has 4 target locations $T1,T2,T3,T4$,
which we will show are reached with probability $\frac{1}{2}$ from the
start location $E0$ of $GetProb$ iff $t= \frac{1}{2^{c_1+1}}$. Thus,
in this case when $t=\frac{1}{2^{c_1+1}}$, we reach $\ell_j$ with the
values $x_1=\frac{1}{2^{c_1+1}}$, $x_2=\frac{1}{2^{c_2}}$, $x_3=x_4=0$
which implies that $c_1$ has been incremented correctly according to
our encoding. We now look at the gadget $GetProb$.

\begin{lemma}
  \label{lem:proof-gadget} For any value $\epsilon\in
  (-\frac{1}{2^{c_1+1}},\frac{1}{2^{c_1+1}})$, the probability to
  reach a target location in $GetProb$ from $E0$ is
  $\frac{1}{2}(1-4\epsilon^2)$ ($\le \frac{1}{2}$). Further this
  probability is equal to $\frac{1}{2}$ iff $\epsilon=0$.
 \end{lemma}
\begin{proof}
  Note that when the start location $E0$ of $GetProb$ is reached, we
  have $x_1=\frac{1}{2^{c_1+1}}+\epsilon$, $x_2=0$, $x_3=
  1-\frac{1}{2^{c_1+1}}+\epsilon$, $x_4=0$. A total of 2 time units
  can be spent at $E0$. It can be seen that transitions to $E3$ and
  $E4$ are
  respectively enabled with the time intervals $[0, 1-
  \frac{1}{2^{c_1+1}}-\epsilon]$ and
  $[1,1+\frac{1}{2^{c_1+1}}-\epsilon]$. Similarly, reaching $E1$ and
  $E2$ are enabled by the time intervals
  $[1-\frac{1}{2^{c_1+1}}-\epsilon,1]$ and
  $[1+\frac{1}{2^{c_1+1}}-\epsilon,2]$. The sum of probabilities of
  reaching either $E3$ or $E4$ is thus $\frac{1}{2}(1-2\epsilon)$.
  Similarly, the sum of probabilities for reaching $E1$ or $E2$ is
  $\frac{1}{2}(1+2\epsilon)$. The locations $P1, P2$ are then reached
  with
  the values $x_1=\frac{1}{2^{c_1+1}}+\epsilon$, $x_2=0$, $x_3=
  1-\frac{1}{2^{c_1+1}}+\epsilon$, $x_4=0$. The probability of
  reaching
  the target locations $T3$ or $T4$ (i.e., through $P1$) from $E0$ is
  hence
  $\frac{1}{2}(1+2\epsilon)\frac{1}{2}(1-2\epsilon)=\frac{1}{4}(1-4\epsilon^2)$,
  while the probability of reaching a target location $T1$ or $T2$
  (i.e., through $P2$) from $E0$ is
  $\frac{1}{2}(1+2\epsilon)\frac{1}{2}(1-2\epsilon)=\frac{1}{4}(1-4\epsilon^2)$.
  Thus, the probability of reaching a target location (one of
  $T1,T2,T3,T4$) in $GetProb$ is,
  $\frac{1}{2}(1-4\epsilon^2)$,
  which is always $\le \frac 1 2$.
   This completes the
  first statement of the lemma. Further, from the expression, we
  immediately have that the probability to reach a target location in
  $GetProb$ from $E0$ is $\frac{1}{2}$ iff $\epsilon=0$.
  \end{proof}
  
  The decrement $c_1$, increment $c_2$ as well as decrement $c_2$
  modules are similar and these as well as the zero test modules can
  be found in the Appendix. %\cite{TR16}.
  
  \begin{lemma} 
    Player~$\Diamond$ has a strategy to reach the (set of)
    target locations in $\cal{G}$ with probability $\frac{1}{2}$ iff
    the two-counter machine does not halt. 
  \end{lemma} 

\begin{proof}
  Suppose the two-counter machine halts (say in $k$ steps). Then there
  are two cases: (a) the simulations of all instructions are correct
  in ${\cal G}$. In this case, the target location can be reached in
  either of the first $k$ steps. By Lemma~\ref{lem:proof-gadget}, the
  probability of reaching a target location in the first $k$ steps is
  the summation $\frac{1}{2}.\frac{1}{2}+(\frac{1}{2})^2.\frac{1}{2} +
  (\frac{1}{2})^3.\frac{1}{2}+\dots+(\frac{1}{2})^k.\frac{1}{2} <
  \frac{1}{2}$. (b) Player~$\Diamond$ made an error in the computation
  in the
  first $k$ steps. But then again by Lemma~\ref{lem:proof-gadget}, the
  finite sum obtained is $< \frac{1}{2}$ (since in the error step(s),
  the probability to reach target locations is
  $\frac{1}{2}-4\epsilon^2<\frac{1}{2}$). Thus, if the two-counter
  machine halts, under any strategy of $\Diamond$ player, the
  probability to reach the target locations is $<\frac{1}{2}$.

  On the other hand, suppose the two-counter machine does not halt.
  Then, if Player~$\Diamond$ chooses the strategy which faithfully
  simulates all instructions of the two-counter machine, the
  probability to reach the (set of) target locations is given by the
  infinite sum $\sum_{i=0}^\infty
  (\frac{1}{2})^{i}\frac{1}{2}=\frac{1}{2}$. Any other strategy of 
  Player~$\Diamond$ corresponds to performing at least one error in the
  simulation. In this case, the infinite sum obtained has at least one
  term of the form $(\frac{1}{2})^k(\frac{1}{2}-4\epsilon^2)$, for
  $\epsilon^2 >0$. Clearly, such an infinite sum does not sum to
  $\frac{1}{2}$.  This concludes the proof.
\end{proof}

The previous proof can be changed for other
  thresholds and to use unbounded intervals and exponential
  distributions.

%% file: two-and-half-new.tex
In this section, we tackle the \emph{time-bounded} version of the quantitative
reachability problem. This strengthens the definition of reachability
 by considering a given time bound $\Delta$, and
requiring that $\mathcal{P}_{\sigma}(\{\rho \in
Run(\mathcal{G},s_0,\sigma) \mid \rho\ \text{visits}\
T\ \text{within}\ \Delta\ \text{time units}) \bowtie p$. 

In this new framework, we show the undecidability of the quantitative
reachability problem for $2 \frac{1}{2}$ STGs. We reduce from the \emph{halting}
problem for two-counter machines (unlike in the previous section, where our
reduction was from the \emph{non-halting} problem), using Player~$\Box$ to
verify the correctness of the simulation. The complication here is that the
total time spent should be bounded and hence we cannot allow arbitrary time
elapses. We will in fact show a global time bound of $\Delta=5$ for this
reduction. 

\begin{theorem}
  \label{thm:undec-two}
  The time-bounded quantitative reachability
  problem is undecidable for $2\frac{1}{2}$ STGs with  $\ge 5$
  clocks.
\end{theorem} 
\begin{proof}
  Let $\mathcal{M}$ be a two-counter machine. We construct an STG with
  5 clocks such that the two-counter machine $\mathcal{M}$ halts iff
  Player $\Diamond$ has a strategy to reach some desired locations
  with probability $\frac{1}{2}$, whatever Player $\Box$ does, and such
  that the total time spent is bounded by $\Delta=5$ units. 

  The main idea behind the proof is that the total time spent in the
  simulation of the $k^{th}$ instruction will be $\frac{1}{2^k}$. We
  thus
  get a decreasing sequence of times $\frac{1}{2}$, $\frac{1}{4}$,
  $\frac{1}{8}\dots$ for simulating the instructions $1,2 \dots$ and
  so on. In total, we will use five clocks $x_1,x_2,z,a$ and $b$. The clocks
  $x_1$ and  $x_2$ are used encode the counter values (along with the current
  instruction number) such that at the end of the $k^{th}$ instruction, if $k$
  is even the values are encoded in $x_1$ and if $k$ is odd they are encoded in
  $x_2$ as follows: 
  \begin{description}
  \item[$(enc_{x_1})$] $k$ is even and $x_1=\frac{1}{2^{k+c_1}3^{k+c_2}}$,
    $x_2=0$, $z=1-\frac{1}{2^k}$, $a=b=0$; 
  \item[$(enc_{x_2})$] $k$ is odd and $x_2=\frac{1}{2^{k+c_1}3^{k+c_2}}$,
    $x_1=0$, $z=1-\frac{1}{2^k}$, $a=b=0$; 
  \end{description}  
  We start the simulation with $x_1=1,x_2=z=0=a=b$ corresponding to the initial
  instruction ($k=0$) and the fact that the values of $C_1, C_2$ are $0$.
  Moreover, if $x_1=\frac{1}{2^{k+c_1}3^{k+c_2}}$  at the end of the $k$th
  instruction, and if the $(k+1)$th instruction is an increment $C_1$
  instruction,  then at the end of the $(k+1)$th instruction, $x_2=
  \frac{1}{2^{k+c_1+2}3^{k+c_2+1}}$. Clock $z$ keeps a separate track of the
  number of instructions simulated so far, by having a value $1-\frac{1}{2^k}$
  after completing the simulation of $k$ instructions.  Clocks $a$ and $b$ are
  auxiliary clocks that we need for the simulation. We assume uniform
  distribution over delays in probabilistic locations. 
If no weight is written on an edge, it is assumed to be 1. 

 We outline the simulation of a increment instruction <<~$\ell_i$ :
increment counter $C_1$, goto $\ell_j$~>> in Figure \ref{incmod},
assuming this is the $(k+1)$th instruction, where $k$ is even. Thus,
at the end of the $k$ first instructions, we have
$x_1=\frac{1}{2^{k+c_1}3^{k+c_2}}$, $z=1-\frac{1}{2^k}$ and
$a=b=x_2=0$ (the other case of odd $k$, i.e., $(enc_{x_2})$ encoding
is symmetric).  At the end of this $(k+1)$th instruction's simulation,
the value of clock $z$ should be $z = 1 - \frac{1}{2^{k+1}}$ to mark
the end of the $(k+1)^{th}$ instruction. Also, we must obtain
$x_2=\frac{x_1}{2^2\cdot 3}=\frac{x_1}{12}$, marking the successful
increment of $C_1$.

\begin{figure} [h]                                                                                                                                                          
 \centering
   \scalebox{0.7}{                                                                                                                                                       
   \begin{tikzpicture}[->,thick]                                                                                                                                          
    \node[initial, dia,label=above:{},initial text={$a,b,x_2=0$}] at (0,-1) (A) 
    {$\ell_i$} ;  
       \node[dia] at (2,-1) (B) {B};
    \node[box,label=below:{$b= 0$}] at (4,-1) (C) {Check};
       \node[dia] at (9,-1) (D) {$\ell_j$};
         \node[rounded rectangle,fill=gray!20!white] at (2,-3) (chkz) {Check $z$};
         \node[rounded rectangle,fill=gray!20!white] at (6,-3) (chkx) {Check $x_2$};
    \path (A) edge node [midway,above] {$a<1$}(B);      
    \path (A) edge node[midway,below] {$x_2:=0$}(B);
    \path (B) edge node[midway,above] {$a<1$}(C);
    \path (B) edge node[midway,below] {$b:=0$}(C);
    \path (C) edge node[] {}(chkz);
    \path (C) edge node[midway,below] {$x_1,a:=0$}(D);
   
    \path (C) edge node[] {}(chkx);
   \end{tikzpicture}                                                                                                                                                      
  }                                                                                                                                                                     
  \caption{Module for incrementing $C_1$ (after an even number of
    steps)}                                                                                                                             
   \label{incmod}                                                                                                                                                        
\end{figure}
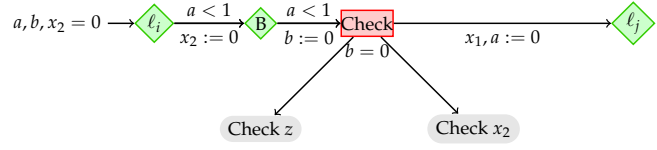 

Player $\Diamond$ elapses times $t_1,t_2$ in locations $\ell_i, B$.
When the player $\Box$ location $Check$ is reached, we have
$a=t_1+t_2=t$ and $x_2=t_2$, $z=1-\frac{1}{2^k}+t_1+t_2$. Player
$\Box$ has three possibilities : (1) to continue the simulation going
to $\ell_{k+2}$, (2) verify that
$t_2=\frac{1}{2^{k+c_1+2}3^{k+c_2+1}}$ by going to the widget `Check
$x_2$' or (3) verify that $t_1+t_2=\frac{1}{2^{k+1}}$ by going to the
widget `Check $z$'. These widgets are given in Figure \ref{check1}.
The probability of reaching a target location in widget `Check $z$' is
$\frac{1}{2}(1-t)+\frac{1}{4}\frac{1}{2^k}=\frac{1}{2}$ iff $t=
\frac{1}{2^{k+1}}$. In widget `Check $x_2$', the transitions from $F1$
to $C1$ and $F1$ to $C2$ are taken with probability $\frac{1}{12}$ and
$\frac{11}{12}$, respectively since the weights of edges connecting
F1,C1 and F1,C2 are respectively 1 and 11.  With this, for
$n=\frac{1}{2^{k+c_1}3^{k+c_2}}$, the probability of reaching a target
location in `Check $x_2$' is
$\frac{1}{2}(1-t_2)+\frac{n}{24}=\frac{1}{2}$ iff $t_2=\frac{n}{12}$.

 \begin{figure} [h]                                                                                                                                                          
   \begin{center}                                                                                                                                                         
   \scalebox{0.7}{                                                                                                                                                       
   \begin{tikzpicture}[->,thick]                                                                                                                                          
                                                                                                                                                                          
    \node[initial, cir,label=above:{$b\equal 0$},initial text={}] at (0,-1) (A) {A0} ;  
    \node[cir,label=below:{$b\leq 1$}] at (2,-1) (B) {B0};
    \node[cir,accepting] at (4,-1) (C) {C0};
    \node[cir] at (2,0.5) (C1) {};
    \path (A) edge node [midway,above] {}(B);
              \path (B) edge node[midway,above] {$a \leq 1$}(C);
              \path (B) edge node[midway,left]{$a>1$}(C1);
    
    \node[cir,label=left:{$b=0$}] at (0,-2) (D) {D0};
    \node[dia] at (1,-2) (D1) {};
    \path (A) edge node[midway,above] {}(D);
    \path (D) edge node[midway,above] {}(D1);
    \node[dia] at (0,-3) (E) {E0};
    \node[cir,label=above:{$b\leq 1$}] at (2,-3) (F) {F0};
    \node[dia] at (2,-4.5) (F1) {};
    \path (D) edge node[midway,above] {}(E);
    \path (E) edge node[midway,above] {$a=1?$}(F);
    \path (E) edge node[midway,below] {$b:=0$}(F);
    \node[cir,accepting] at (4,-3) (G) {G0};
    \path (F) edge node[midway,above] {$z \leq 2$}(G);
    \path (F) edge node[midway,left] {$z > 2$}(F1);
    
   \node[initial, cir,label=above:{$b\equal 0$},initial text={}] at (7,-1) (A1) {A1} ;
    \node[cir,label=below:{$b\leq 1$}] at (9,-1) (B1) {B1};
    \node[cir,label=left:{$b=0$}] at (7,-3) (B2) {F1};
    \path (A1) edge node[midway,above] {}(B1);
    \path (A1) edge node[midway,above] {}(B2);
    \node[cir,accepting] at (11,-1) (B3) {};
    \node[cir] at (9,0.5) (B4) {};
\path (B1) edge node[midway,above] {$x_2 \leq 1$}(B3);
    \path (B1) edge node[midway,left] {$x_2>1$}(B4);
    
\node[dia] at (7,-4) (C1) {C2};
    \path (B2) edge node[midway,left] {11}(C1);
\node[dia] at (8,-3) (C2) {C1};
    \path (B2) edge node[midway,above] {1}(C2);
\node[dia] at (10,-3) (C3) {D1};
    \path (C2) edge node[midway,above] {$a=1?$}(C3);
\path (C2) edge node[midway,below] {$a:=0$}(C3);
\node[cir,label=above:{$b\leq 1$}] at (12,-3) (C4) {E1};
\path (C3) edge node[midway,above] {$x_1=2$}(C4);
\path (C3) edge node[midway,below] {$b:=0$}(C4);
\node[cir,accepting] at (14,-3) (C5) {};
\node[dia] at (12,-4.5) (C6) {};
\path (C4) edge node[midway,above] {$a \leq 1$}(C5);
\path (C4) edge node[midway,left] {$a>1$}(C6);
    \end{tikzpicture}                                                                                                                                                      
  }                                                                                                                                                                     
  \caption{Widgets `Check $z$' (left) and `Check $x_2$' (right)}
   \label{check1}                                                                                                                                                        
  \end{center}                                                                                                                                                           
\end{figure}
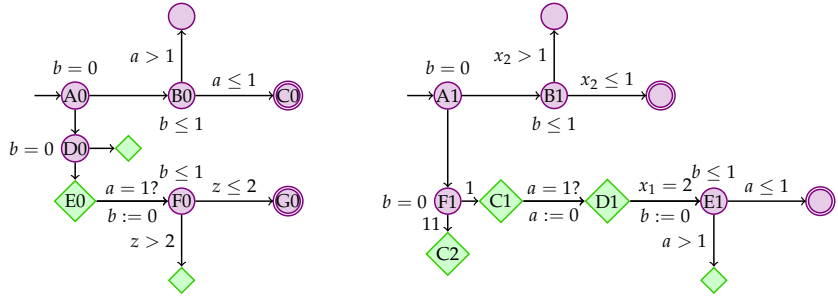 

\noindent{\bf Time elapse for Increment.}
If player $\Box$ goes ahead with the simulation, the time elapse for
the $(k+1)$th instruction is $t_1+t_2=\frac{1}{2^{k+1}}$.  Consider
the case when player $\Box$ goes in to `Check $z$'. The time elapse
till now is $\frac{1}{2}+\dots +\frac{1}{2^{k+1}}$. The time spent in
the `Check $z$' widget is as follows: one unit is spent at location
$B0$, one unit at location $F0$, and $1-t$ units at location
$E0$. Thus, $\leq 3$ units are spent at the `Check $z$'
widget. Similarly, the time spent in the `Check $x_2$' widget is one
unit at $B1$, $1-t$ units at $C1$, $1-n$ units at $D1$ and one unit at
$E1$. Thus a time $\leq 4$ is spent in `Check $x_2$'.
Thus, the time spent till the $(k+1)$th instruction is $\leq
\frac{1}{2}+\dots \frac{1}{2^{k+1}}+4$ if player $\Box$ goes in for a
check, and otherwise it is $\frac{1}{2}+\dots +\frac{1}{2^{k+1}}$.

\noindent{\bf Other increment, decrement, zero-check Instructions.}
The main module corresponding to \emph{increment $C_2$} and decrement
$C_1, C_2$ is the same as in Figure \ref{incmod}. The only change
needed is in the `Check $x_2$' widget. While incrementing $c_2$, we
need $x_2=\frac{x_1}{2\cdot 3^2}=\frac{x_1}{18}$. This is done by
changing the weights on the outgoing edges from $F1$ to $C1$ and $C2$
to $1$ and $17$ respectively. Similarly, while \emph{decrementing
  $C_1$}, we need $x_2=\frac{x_1}{3}$. This is done by changing the
weights on the outgoing edges of $F1$ to $1,2$ respectively. Lastly,
to \emph{decrement $C_2$}, we need $x_2=\frac{x_1}{2}$, and in this
case the weights are 1 each.

The zero check module is a bit more complicated. The broad idea is
that we use a diamond node to guess whether the current clock (say
$C_1$) value is zero and branch into two sides (zero and
non-zero). Then we use a box node on each branch to verify that the
guess was correct. If correct, we proceed with the next instruction,
if not, we check this by going to a special widget. In this widget, we
can reach a target node with probability $\frac{1}{2}$ iff the guess
is correct. The details of this widget and the proof that all these
simulations can be done in time bounded by $\Delta\leq 5$ units
is given in  the Appendix. 
 \end{proof}

%% file: new-quant-dec.tex
We have seen in the previous section that the quantitative
reachability problem is undecidable in $1\frac{1}{2}$ STGs with $\ge
4$ clocks. In this section we study the \emph{quantitative
  reachability} problem in the setting of $1\frac{1}{2}$ STGs
\textit{with a single clock}. In~\cite{qest08}, the quantitative
reachability problem in $\frac{1}{2}$ STGs with a single clock, under
certain restrictions, was shown to be decidable by reducing it to the
quantitative reachability problem for finite Markov chains. In our
case, we lift this to $1\frac{1}{2}$ STGs with a single clock, under
similar restrictions, by reducing to the quantitative reachability
problem in finite Markov decision processes (MDPs in short).

For the rest of this section, we consider a $1\frac{1}{2}$ STG
$\stg=(\mathcal{A},(L_{\Diamond},L_{\bigcirc}),\omega,\mu)$ with a
single clock denoted $x$. We write $c_{\max}$
  for the maximal constant
appearing in a guard of $\stg$.

We assume w.l.o.g. that target locations belong to player
$\Diamond$ (a slight modification of the construction can be done
  if this is not the case). In the following, when we talk about
regions, we mean the clock regions from the classical region
construction for timed automata~\cite{AD94,LMS04}: since $\stg$
  has a single clock, regions in this case are simply either
  singletons $\{c\}$ with $c \in \mathbb{Z}_{\ge 0} \cap [0;c_{\max}]$, or
  open intervals $(c,c+1)$ with $c \in \mathbb{Z}_{\ge 0} \cap
  [0;c_{\max}-1]$, or the unbounded interval $(c_{\max};+\infty)$.  While region
  automata are standardly finite automata, we build here from $\stg$ a
  \emph{region STG} $\stg_{\mathcal{R}}$, which has only clock
  constraints defined by regions (that is, either $x=c$ or $c<x<c+1$
  or $x>c_{\max}$), and such that each location of $\stg_{\mathcal{R}}$ is
  indeed a pair $(\ell,R)$ where $\ell$ is a location of $\stg$ and
  $R$ a region (region $R$ is for the region which is hit when
  entering the location). While it is not completely standard, this
  kind of construction has been already used
  in~\cite{journal-sta,qest08,icalp09}, and questions asked on $\stg$
  can be equivalently asked (and answered) on $\stg_{\mathcal{R}}$.
  Now, we make the following restrictions on $\stg_{\mathcal{R}}$
  (which yields restrictions to $\stg$), which we denote $(\star)$:
\begin{enumerate}
\item The TA $\mathcal{A}$ is assumed to be structurally
 non-Zeno: any bounded cycle of $\mathcal{A}$ (a cycle
 in which all edges have a non-trivial upper-bound) contains at least
 one location whose associated region is the zero region (i.e., edge
 leading to it, resets the clock).
\item For every state $s = ((\ell,r),\nu)$ of
   $\stg_{\mathcal{R}}$ such that $\ell \in L_\bigcirc$,
   $I(s)=\Rplus$, and $\mu_s$ is an exponential distribution; Furthermore
   the rate of $\mu_s$ only depends on location $\ell$.
\item $\stg_{\mathcal{R}}$ is \emph{initialized}, that is, any edge
 from a non-stochastic location to a stochastic location resets the
 clock $x$.
\end{enumerate}
While the first two assumptions are already made in~\cite{qest08},
even in the $\frac{1}{2}$ player case, the third condition is new. In
the following we denote $\mathbf{0}$ for the region $\{0\}$ and
$\infty$ for the unbounded region $(c_{\max};+\infty)$.

We now show how to obtain an MDP from the STG $\stg_{\Rr}$. The
construction is illustrated on Figure \ref{ex1}.

A node $(\ell,R)$ of $\stg_{\Rr}$ with $\ell \in L_\bigcirc$ is
\emph{deletable} if $R$ is neither the region $0$ nor the region
$\infty$.
In Figure \ref{ex1}, $(B, (0,1))$ and $(A,(0,1))$ in $\stg_\Rr$ are
what we call deletable nodes. Then, we recursively remove all
deletable nodes $\stg_\Rr$ while labelling remaining paths with
(finite) sequences of edges; each surviving edge is labelled by the
probability of the (provably) finitely many sequences of edges
appearing in the label. One can prove that this object is actually an
MDP, which we denote $M_{\stg}$. Target states in $M_\stg$ are defined
as the pairs $(\ell,R)$ where $\ell$ is a target location in $\stg$.
We can prove that:
\begin{lemma}
  \label{lem:quant-main}
  If $\stg$ is an $1 \frac{1}{2}$ player STG with one clock satisfying
  the hypotheses $(\star)$,
  then $M_{\stg}$ is an MDP such that: (a)
for every strategy $\lambda_\Diamond$ of player $\Diamond$ in
    $\stg$, we can construct a strategy $\sigma_\Diamond$ of player
    $\Diamond$ in $M_{\stg}$ such that the probability of reaching a
    target location in $\stg$ is the same as the probability of
    reaching a target state in $M_{\stg}$; and (b)
for every strategy $\sigma_\Diamond$ of player $\Diamond$ in
    $M_\stg$, we can construct a strategy $\lambda_\Diamond$ of player
    $\Diamond$ in $\stg$ such that the probability of reaching a
    target location in $M_\stg$ is the same as the probability of
    reaching a target state in $\stg$.
\end{lemma}
This lemma allows to reduce the quantitative reachability problem from
the $1 \frac{1}{2}$ STG $\stg$ to the MDP $M_{\stg}$. 

\begin{figure} [t]                                                                                                                                                          
  \centering                                                                            \hspace*{-5mm}                                                                           
   \scalebox{0.7}{                                                                                                                                                       
   \begin{tikzpicture}[->,thick]                                                                                                                                          
    \node[initial, cir, initial text ={}] at (-17,0) (A) {A} ;                                                                                                                      
    \node[cir] at (-13,0) (B) {B}; 
    \node[cir] at (-17,-2) (C) {C};  
    \node[dia] at (-13,2) (D) {D};
    \node[dia] at (-13,-2) (E) {E};
    \path (A) edge node [above]{$x < 1$}(B);
    \path (A) edge node [below]{$e_4$}(B);
    \path(C) edge [bend left] node [left]{$\begin{array}{c}e_3\\x < 1\end{array}$}(A); 
    \path (A) edge[bend left] node [right]{$\begin{array}{c}e_1\\x\geq 1\\x:=0 \end{array}$}(C);
    \path (C) edge node [below]{$\begin{array}{c}e_2\\x \geq 1\end{array}$}(E);
        \path (B) edge node [left]{$x < 1$} (D);
         \path (B) edge node [right]{$e_7$} (D);
        \path (D) edge[bend right] node [left]{$e_8,x < 1$}node[below]{$x:=0$} (A);
    \path (B) edge[bend left] node [right]{$\begin{array}{c}e_5\\x\geq 1\\x:=0 \end{array}$}(E);
    \path (E) edge [bend left] node[left]{$\begin{array}{c}e_6\\x < 1\\x:=0\end{array}$}(B);
    
 \node[initial, cir, initial text ={}] at (-9,2) (A0) {A, $\mathbf{0}$} ;                                                                                                                      
    \node[cir] at (-7,2) (B01) {B,(0,1)}; 
   \node[dia] at (-5,2) (D01) {D,(0,1)};
     \node[cir] at (-9,0) (C0) {C,$\mathbf{0}$};
      \node[cir] at (-9,-2) (A01) {A,(0,1)};
      \node[dia] at (-7,0) (E0) {E,$\mathbf{0}$};
      \node[cir] at (-5,0) (B0) {B,$\mathbf{0}$};
      \node[dia] at (-7,-2) (E1) {E,$\infty$};
      
     \path(A0) edge node[above]{$e_4$}(B01);
     \path(B01) edge node[above]{$e_7$}(D01);
     \path(D01) edge[bend right] node[above]{$e_8$}(A0);
     \path(A0) edge node[left]{$e_1$}(C0);
      \path(C0) edge node[right]{$e_3$}(A01); 
      \path(A01) edge[bend left] node[left]{$e_1$}(C0); 
      \path(C0) edge node[above]{$e_2$}(E1); 
      \path(A01) edge node[left]{$e_4$}(B01);
      \path(B01) edge node[right]{$e_5$}(E0); 
        \path(E0) edge node[above]{$e_6$}(B0); 
        \path(B0) edge[bend left] node[above]{$e_5$}(E0); 
      \path(B0) edge node[right]{$e_7$}(D01);

 \node[initial, cir, initial text ={}] at (-1,0) (A0) {A, $\mathbf{0}$} ;                                                                                                                      
   \node[dia] at (-1,2) (D01) {D,(0,1)};
     \node[cir] at (-1,-2) (C0) {C,$\mathbf{0}$};
      \node[dia] at (1,-2) (E0) {E,$\mathbf{0}$};
      \node[cir] at (3,-2) (B0) {B,$\mathbf{0}$};
      \node[dia] at (-3,-2) (E1) {E,$\infty$};
       
       \path(A0) edge[draw=black,line width=1.5pt,->,double=white, bend right] node[right]{$e_4e_7$}(D01);
       \path(D01) edge[draw=black,line width=1.5pt,->,double=white, bend right] node[right]{$e_8$}(A0);
      %%1-2e^-1
       \path(A0) edge[draw=black,line width=1.5pt,->,double=white] node[right]{$e_4e_5$}(E0);
      %%e^-1
      \path(A0) edge[draw=black,line width=1.5pt,->,double=white] node[right]{$e_1$}(C0);
      %%e^-1
      \path(B0) edge[draw=black,line width=1.5pt,->,double=white, bend right] node[left]{$e_7$}(D01);
      %%1-e^-1
      \path(B0) edge[draw=black,line width=1.5pt,->,double=white] node[above]{$e_5$}(E0);
      %%e^-1
      \path(E0) edge[draw=black,line width=1.5pt,->,double=white,bend right] node[below]{$e_6$}(B0);
      %%1-e^-1
      \path(C0) edge[draw=black,line width=1.5pt,->,double=white,loop below] node[below]{$e_3e_1$}(C0);
      %%e^-1
      \path(C0) edge[draw=black,line width=1.5pt,->,double=white] node[above]{$e_2$}(E1);
      %%e^-1
      \path(C0) edge[draw=black,line width=1.5pt,->,double=white,bend left=50] node[left]{$e_3e_4e_7$}(D01);
            %%1-3e^-1+e^-2+ int(0,1)t e^{-t-1}dt
\path(C0) edge[draw=black,line width=1.5pt,->,double=white] node[above]{$e_3e_4e_5$}(E0);
       %%e^-1-e^-2- int(0,1)t e^{-t-1}dt
                   
                                                         \end{tikzpicture}                                                                                                                                                      
  }                                                                                                                                                                     
 \caption{An initialized  $1 \frac{1}{2}$ player STG $\stg$, its region
   game graph $\stg_\Rr$ and the MDP abstraction $M_{\stg}$. 
 }
   \label{ex1}                                     
\end{figure}
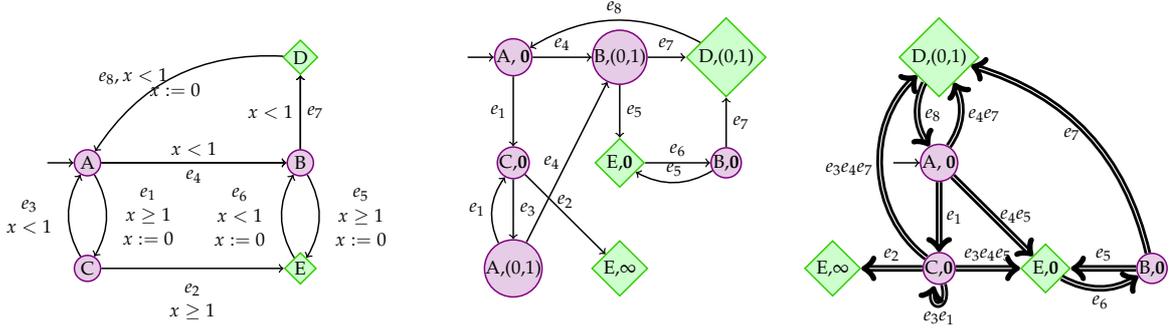

As an example, in Figure \ref{ex1}, we show a $1 \frac{1}{2}$ player
STG $\stg$, its region game graph $\stg_\Rr$ (guards omitted for
readability) and the MDP abstraction $M_{\stg}$. Note that all
$\Diamond$ nodes remain, while only those stochastic nodes with
regions $\mathbf{0}$ and $\infty$ are retained in $M_{\stg}$. The
stochastic nodes $(B,(0,1))$ as well as $(C,(0,1))$ are deleted in
$M_{\stg}$. On deleting nodes from the region graph, the probability
on the edges of $M_{\stg}$ is the probability of the respective paths
from the region graph. For example, the edge from $(A,0)$ to
$(D,(0,1))$ is labelled with $e_4e_7$ by deleting $(B,(0,1))$.

Thus, the remaining thing that has to be addressed now is how to compute the probabilties and compare them with a rational threshold.  The first thing to note is that the edges of the MDP are all labelled with polynomials over exponentials obtained using the delays from the underlying game with rational coefficients. For example, in Figure~\ref{ex1}, in the MDP in the rightmost picture, we obtain:
\newcommand{\cP}{\mathcal{P}}
$\cP (e_1){=}\cP (e_2) {=} \cP(e_5){=} e^{-1}$, $\cP (e_6){=} \cP(e_7)
{=} \cP(e_8){=}1{-}e^{-1}$, 
$\cP (e_4 e_5){=} e^{-1} {-} e^{-2}$, $\cP (e_4 e_7){=} 1{-} 2e^{-1}$, 
$\cP (e_3 e_4 e_7) {=}  2 {-} 5e^{-1} {+} e^{-2}$, $\cP (e_3 e_4 e_5){=} 1{-}
e^{-1} {+} e^{-2}$, and $\cP (e_3 e_1){=} \frac{1}{2} (1{-}e^{-2})$.
It can be seen that we can write each of these probabilities as a
polynomial in $e^{-1}$. More generally, for any MDP with differing
rates (of the exponential distribution) in each state, we get a set of
rational functions in $e^{-\frac{1}{q}}$ for some
$q\in\mathbb{Z}_{>0}$, where $q$ is obtained as a function of the
rates in each state. Thus, using standard algorithms for
MDPs~\cite{BK}, and as done for Markov chains in~\cite{qest08}, we get
that we can compute expressions for the probability of reaching the
targets, and decide the threshold problem.

\begin{theorem}
\label{thm:quant-dec}
Quantitative reachability for 1-clock $1\frac{1}{2}$-player STGs
satisfying $(\star)$ is decidable.
\end{theorem}

We can lift this construction to include $\Box$ player nodes, keeping
the same initialized restriction with $\Box$ nodes as well.  Then the
region game graph $\stg_\Rr$ includes $\Box$ nodes in the obvious way,
and we consider strategy profiles of $\Box$ and $\Diamond$. The
question then is to check if $\Diamond$ has a strategy to reach a
target with probability $\sim c$ against all possible strategies of
$\Box$ in $M_{\stg}$. Hence we have that
\begin{corollary}
\label{cor:dec}
Quantitative reachability for 1-clock $2 \frac{1}{2}$ player STGs
satisfying $(\star)$ is decidable.
\end{corollary}

%% file: discussion.tex
In this paper, we have  
refined the decidability boundaries for STGs as summarized in the table in Introduction. The significance of our undecidability results for quantitative reachability (via different two-counter machine reductions) lies in the fact that they introduce ideas which could potentially help in settling other open problems. We highlight these below:

\begin{itemize}
\item for $1\frac{1}{2}$ player games, the crux is to cleverly encode the error $\epsilon$ made by player $\Diamond$ in such a way that it reflects as $\frac{1}{2}-\epsilon^2$ in the resulting probability. This ensures that  the $\Diamond$ player can never cheat and the probability will be $< \frac{1}{2}$ as 
soon as there is an error (even when simulating a non-halting run of the two-counter machine). Indeed, this is why the reduction is from the non-recursively enumerable non-halting problem. 
\item  
  for $2\frac{1}{2}$ player games in the \emph{time-bounded setting}, we obtain undecidability by showing a reduction from halting problem for two-counter machines. This is surprising, as 
  time-boundedness restores decidability in several classical undecidable problems like the inclusion problem in timed automata~\cite{ORW09,OW10}. In the case of priced timed games \cite{concur14}, time-boundedness gives undecidability; however, this can be attributed to  
  the fact that price variables are not clocks, and can grow at different rates in different locations. Somehow, the combination of simple clocks and probabilities achieves the same. 
\end{itemize}
Combining these ideas would, e.g., allow us to improve Theorem \ref{thm:undec-two} by showing undecidability of time bounded, quantitative reachability in $1\frac{1}{2}$ player STGs with a larger number of clocks. 
The main intricacy is to replace $\Box$ player nodes by stochastic nodes, and adapt the gadgets in such a way that, within a global time bound, the probability of reaching a target is $\frac{1}{2}$ iff all simulations  are correct and the two-counter machine does not halt.  
As another example, if in the first item above, we obtain a probability of $1-\epsilon^2$ (rather than $\frac{1}{2}-\epsilon^2$), this would settle the (currently open) \emph{qualitative} reachability problem for $2\frac{1}{2}$ games~\cite{icalp09}.

Coming to decidability results, we have for the first time characterized a family of $1\frac{1}{2}$,$2\frac{1}{2}$ player STGs for whom the quantitative reachability is decidable.  The use of exponential distributions is mandatory to get a closed form expression for the probability.
It is unclear if this construction can be extended to some larger classes of STGs. Figure 9 in~\cite{journal-sta} shows an example of a two-clock $\frac{1}{2}$ player game for which the region abstraction fails to give any relevant information on the real ``probabilistic'' behaviour of the system (lack of so-called fairness); in particular it cannot be used for qualitative, and therefore quantitative, analysis of reachability properties. 
The decidability of qualitative reachability in $1\frac{1}{2}, 2\frac{1}{2}$, multi-clock STG seems then hard due to the same problem of unfair runs.  If one restricts to one clock, then the qualitative reachability of $1\frac{1}{2}$ STGs is decidable \cite{icalp09}.  We conjecture that this can be extended to $2\frac{1}{2}$ STGs in the single clock case.

%% file: appendix-new.tex
\newpage
\centerline{\Large{\bf Appendix}}

\section{Counter Machines}
A two-counter machine $M$ is a tuple $(L, C)$ where ${L = \set{\ell_0,
    \ell_1, \ldots, \ell_n}}$ is the set of instructions---including a
distinguished terminal instruction $\ell_n$ called HALT---and ${C =
  \set{c_1, c_2}}$ is the set of two \emph{counters}.  The
instructions $L$ are one of the following types:
\begin{enumerate}
\item (increment $c$) $\ell_i : c := c+1$;  goto  $\ell_k$,
\item (decrement $c$) $\ell_i : c := c-1$;  goto  $\ell_k$,
\item (zero-check $c$) $\ell_i$ : if $(c >0)$ then goto $\ell_k$
  else goto $\ell_m$,
\item (Halt) $\ell_n:$ HALT.
\end{enumerate}
where $c \in C$, $\ell_i, \ell_k, \ell_m \in L$.
A configuration of a two-counter machine is a tuple $(l, c, d)$ where
$l \in L$ is an instruction, and $c, d$ are natural numbers that specify the value
of counters $c_1$ and $c_2$, respectively.
The initial configuration is $(\ell_0, 0, 0)$.
A run of a two-counter machine is a (finite or infinite) sequence of
configurations $\seq{k_0, k_1, \ldots}$ where $k_0$ is the initial
configuration, and the relation between subsequent configurations is
governed by transitions between respective instructions.
The run is a finite sequence if and only if the last configuration is
the terminal instruction $\ell_n$.
Note that a two-counter  machine has exactly one run starting from the initial
configuration. 
The \emph{halting problem} for a two-counter machine asks whether 
its unique run ends at the terminal instruction $\ell_n$.
It is well known~(\cite{Min67}) that the halting problem for
two-counter machines is undecidable.

\section{Undecidability of Quantitative Reachability for $1\frac{1}{2}$ STGs}
\label{undec-one-full}
We complete the proof of the undecidability for qualitative reachability in  
$1\frac{1}{2}$ STGs. The simulation of an increment instruction was described in section \ref{one}.
Here we describe the gadgets simulating  decrement and zero test instructions.
Figure \ref{zero-test} describes the gadget simulating the instruction 
$\ell_i$ : If $C_1>0$, then goto $\ell_j$, else goto $\ell_k$.
  It can be seen that with probability $\frac{1}{2}$, the next instruction 
  is simulated, while with probability $\frac{1}{2}$, we reach a target location.

\begin{figure} [th]                                                                                                                                                          
   \begin{center}                                                                                                                                                         
   \scalebox{0.7}{                                                                                                                                                       
   \begin{tikzpicture}[->,thick]                                                                                                                                          
     \node[initial, dia, initial text ={$x_1=\frac{1}{2^{c_1}},x_4=0$}] at (-5,0)  (A) {$\ell_i$} ;                         
        \node[cir,label=above:$x_4\equal 0$] at (-3,2) (B1) {B1}; 
    \node[cir,label=below:$x_4\equal 0$] at (-3,-2) (B2) {B2};
    \node[cir, accepting] at (-3,-4) (C2){T};
    \node[cir, accepting] at (-3,4) (C1){T};
    \node[dia] at (-1,-2) (D2) {$\ell_j$};
    \node[dia] at (-1,2) (D1) {$\ell_k$};
\path (A) edge node [left]{$x_1=1$} (B1);
    \path (A) edge node [left]{$x_1<1$} (B2);
     \path(B2) edge node{} (C2);
    \path(B1) edge node {}(C1);
        \path(B2) edge node {}(D2);
    \path(B1) edge node {}(D1);
   
        \end{tikzpicture}
}
\caption{Zero Test  Instruction}
\label{zero-test}
       \end{center}
    \end{figure}
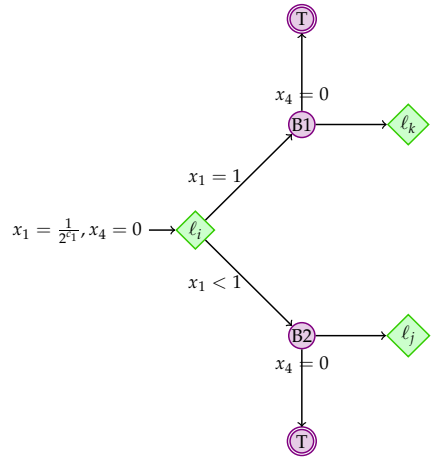

Next, let us see the simulation of a decrement instruction 
$\ell_i$: decrement $C_1$, goto $\ell_j$. Figure \ref{dec-c1} 
depicts this.

\begin{figure} [h]                                                                                                                                                          
   \begin{center}                                                                                                                                                         
   \scalebox{0.7}{                                                                                                                                                       
   \begin{tikzpicture}[->,thick]                                                                                                                                          
                                                                                                                                                                        
    \node[initial, dia, initial text ={$x_1=\frac{1}{2^{c_1}},x_3=0$}] at (-5,0)  (A) {$\ell_i$} ;                                                                                                                      
    \node[cir,label=right:$x_4\equal 0$] at (0,0) (B) {B};
    \node[dia] at (0,2) (D) {D};
\node[dia] at (0,4) (F) {$\ell_j$};
\node[dia] at (0,-2) (F1) {C};
\node[dashed,rectangle,draw=red,fill=yellow] at (0,-4)(E){GetProb}; 
         \path (A) edge node [above]{$0<x_1,x_3<1$} (B);
    \path (A) edge node [below]{$\{x_4\}$} (B);
    \path (B) edge node[left]{$\{x_1\}$} (D);
    \path (B) edge node[left]{$\{x_2\}$} (F1);
    \path(F1) edge node[left] {$x_1=1$} (E);
    \path(F1) edge node[right] {$\{x_1,x_4\}$} (E);
    \path (D) edge[loop left] node [left]{$x_2=1,\{x_2\}$} (D);
   \path (D) edge node [left]{$x_3=1$} (F);
   \path (D) edge node [right]{$\{x_3,x_4\}$} (F);
                   \node[cir,initial,initial where=right, initial text={}, label=left:$x_4 \leq 2$] at (6,0) (E0) {E0};
                   \node[cir,accepting] at (10,-9) (T1) {};
                   \node[cir,accepting] at (2,-9) (T2) {};
                   \node[cir,accepting] at (10,9) (T3) {};
                   \node[cir,accepting] at (2,9) (T4) {};
                   \node[cir] at (8,11) (T30) {};
                   \node[cir] at (3,11) (T40) {};
                   
                   \node[cir] at (7,-10) (T10) {};
                   \node[cir] at (4,-10) (T20) {};

             \node[dia] at (8,-2) (E1) {E1};
                \node[dia] at (4,-2) (E2) {E2};
  \node[dia] at (8,2) (E3) {E3};
                \node[dia] at (4,2) (E4) {E4};
   \node[dia] at (4,-4) (G) {G};
   \node[dia] at (8,-4) (H) {H};
 \node[dia] at (4,-6) (G1) {G1};
   \node[dia] at (8,-6) (H1) {H1};

       \node[dia] at (8,4) (I) {I};
   \node[dia] at (4,4) (J) {J};
       \node[dia] at (8,6) (I1) {I1};
   \node[dia] at (4,6) (J1) {J1};
    \node[cir,label=below:$x_4 \leq 2$] at (6,9) (P1) {P1};
    \node[cir,label=above:$x_4 \leq 2$] at (6,-9) (P2) {P2};                                                 
                
   \path (E0) edge node [right]{$x_3 \geq 1 \wedge x_4 \leq 1$} (E1);
   %%[1/2^i,1]
   \path (E0) edge node [left]{$x_2 \geq 2 \wedge x_4 \leq 2$} (E2);
  %%[2-1/2^i+epsilon, 2]
  \path (E0) edge node [right]{$x_3 \leq  1$} (E3);
  %%[0, 1/2^i]
   \path (E0) edge node [left]{$x_4 \geq 1 \wedge x_2 \leq 2$} (E4);
  %%[1, 2-1/2^i+epsilon]
  \path (E2) edge node [left]{$x_4=2$} (G);
  \path (E2) edge node [right]{$\{x_1,x_4\}$} (G);
  \path (E1) edge node [left]{$x_4=2$} (H);
  \path (E1) edge node [right]{$\{x_1,x_4\}$} (H);
  
  \path (E4) edge node [left]{$x_4=2$} (J);
  \path (E4) edge node [right]{$\{x_1,x_4\}$} (J);
  \path (E3) edge node [left]{$x_4=2$} (I);
  \path (E3) edge node [right]{$\{x_1,x_4\}$} (I);
    \path (G) edge[loop left] node [left]{$x_3=3,\{x_3\}$} (G);                                                                                                                            
                                                                                                                                                                                                                                                                                                                                                                                            \path (H) edge[loop right] node [right]{$x_3=3,\{x_3\}$} (H);
                                                                                                                                                                                                                                                                                                                                                                                             \path (J) edge[loop left] node [left]{$x_3=3,\{x_3\}$} (J);                                                                                                                            
                                                                                                                                                                                                                                                                                                                                                                                            \path (I) edge[loop right] node [right]{$x_3=3,\{x_3\}$} (I);
                                                                                                                                                                                                                                                                                                                                                                                            
                                                                                                                                                                                                                                                                                                                                                                                             \path (G) edge[left] node [left]{$x_2=3$} (G1);                                                                                                                            
                                                                                                                                                                                                                                                                                                                                                                                           \path (G) edge[right] node [right]{$\{x_4,x_2\}$} (G1);

                                                                                                                                                                                                                                                                                                                                                                                              \path (H) edge[left] node [left]{$x_2=3$} (H1);                                                                                                                            
                                                                                                                                                                                                                                                                                                                                                                                           \path (H) edge[right] node [right]{$\{x_4,x_2\}$} (H1);
                                                                                                                                                                                                                                                                                                                                                                                           
                                                                                                                                                                                                                                                                                                                                                                                               \path (I) edge[left] node [left]{$x_2=3$} (I1);                                                                                                                            
                                                                                                                                                                                                                                                                                                                                                                                           \path (I) edge[right] node [right]{$\{x_4,x_2\}$} (I1);
                                                                                                                                                                                                                                                                                                                                                                                           
                                                                                                                                                                                                                                                                                                                                                                                           \path (J) edge[left] node [left]{$x_2=3$} (J1);                                                                                                                            
                                                                                                                                                                                                                                                                                                                                                                                           \path (J) edge[right] node [right]{$\{x_4,x_2\}$} (J1);
                                                                                                                                                                                                                                                                                                                                                                                           
                                                                                                                                                                                                                                                                                                                                                                                           \path (J1) edge[bend left] node [left]{$x_1=1$} (P1);
\path (J1) edge[bend left] node [right]{$\{x_1,x_4\}$} (P1);                                                                                                                                                                                                                                                                                                                                                                                           
                                                                                                                                                                                                                                                                                                                                                                                            \path (I1) edge[bend right] node [left]{$x_1=1$} (P1);
\path (I1) edge[bend right] node [right]{$\{x_1,x_4\}$} (P1);

                                                                                                                                                                                                                                                                                                                                                                                             \path (G1) edge[bend right] node [left]{$x_1=1$} (P2);
\path (G1) edge[bend right] node [right]{$\{x_1,x_4\}$} (P2);                                                                                                                                                                                                                                                                                                                                                                                           
                                                                                                                                                                                                                                                                                                                                                                                            \path (H1) edge[bend left] node [left]{$x_1=1$} (P2);
\path (H1) edge[bend left] node [right]{$\{x_1,x_4\}$} (P2);
                                                                                                                                                                                                                                                                                                                                                                                            
                                                                                                                                                                                                                                                                                                                                                                                            \path(P2) edge node[above]{$x_3 \leq  1$}(T1);
\path(P2) edge node[above]{$x_4 \geq 1 \wedge x_2 \leq 2$}(T2);                                                                                                                                                                                                                                                                                                                                                                                            

\path(P2) edge node[right]{$x_3 \geq 1 \wedge x_4 \leq 1$}(T10);
\path(P2) edge node[left]{$x_2 \geq 2 \wedge x_4 \leq 2$}(T20);

                                                                                                                                                                                                                                                                                                                                                                                             \path(P1) edge node[below]{$x_3 \geq 1 \wedge x_4 \leq 1$}(T3);
                                                                                                                                                                                                                                                                                                                                                                                               \path(P1) edge node[below]{                                                                                                                                                                                                                                                                                                                                          $x_2 \geq 2 \wedge x_4 \leq 2$}(T4);

\path(P1) edge node[right]{$x_3 \leq 1$}(T30);
                                                                                                                                                                                                                                                                                                                                                                                               \path(P1) edge node[left]{                                                                                                                                                                                                                                                                                                                                          $x_4 \geq 1 \wedge x_2 \leq 2$}(T40);
                                                                                                                                                                                                                                                                                                                                                                                            
  \end{tikzpicture}                                                                                                                                                      
  }                                                                                                                                                                     
  \caption{The decrement $c_1$ module on the left and the GetProb gadget on the right}                                                                                                                                       
   \label{dec-c1}                                                                                                                                                        
  \end{center}                                                                                                                                                         
\end{figure}
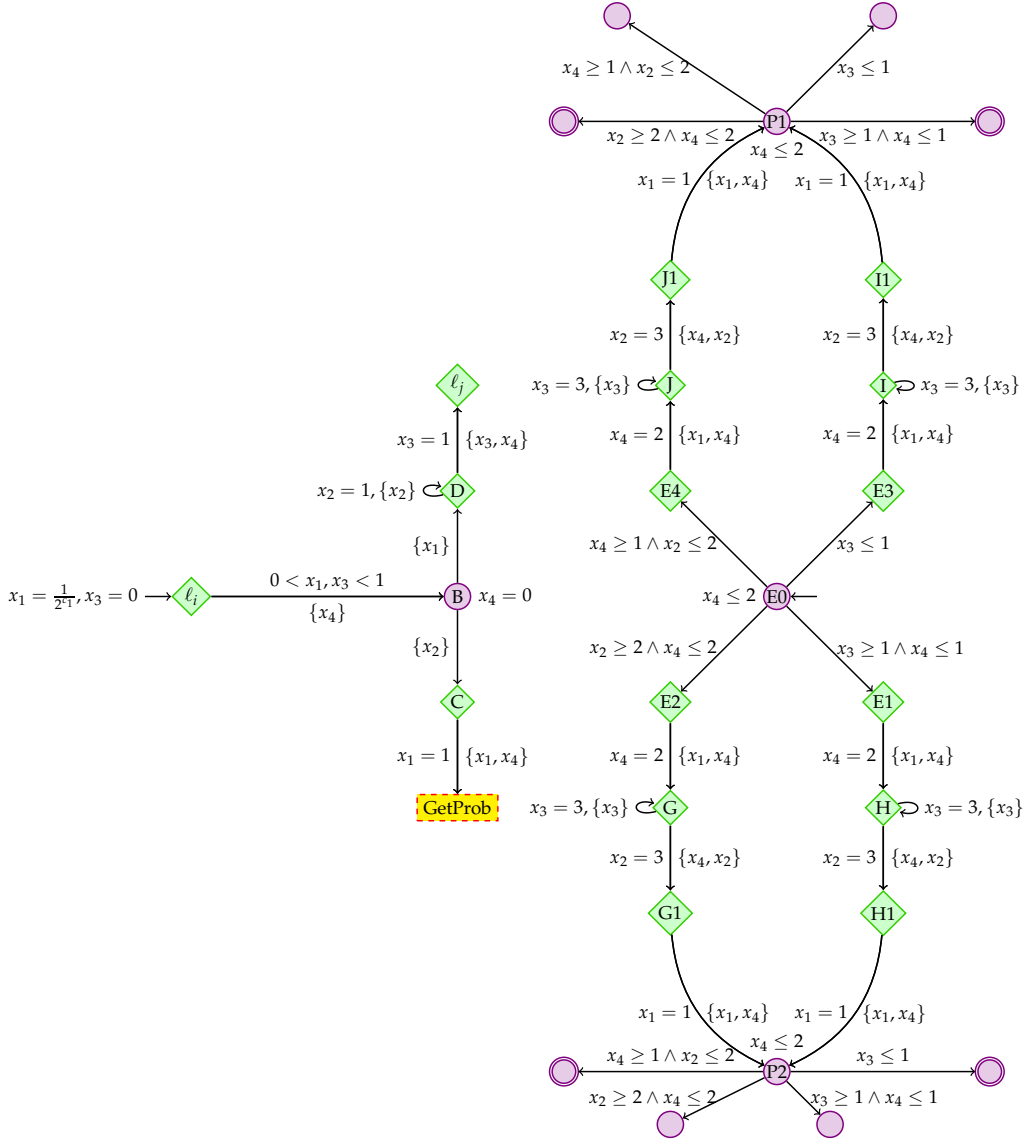 

The decrement module has as its initial location $\ell_i$, which 
is entered with values $x_1=\frac{1}{2^{c_1}}, x_2=\frac{1}{2^{c_2}}, x_3=x_4=0$. 
A non-deterministic time $t$ is spent at $\ell_i$. Ideally, $t=1-\frac{1}{2^{c_1-1}}$. 
At the stochastic node $B$, no time is spent. 
The simulation continues from the location $D$ : $D$ is entered resetting $x_1$.
At $D$ we thus have $x_1=0,x_2=  \frac{1}{2^{c_2}}+t, x_3=t,x_4=0$.
At $D$, a time $1-t$ is spent, reaching $\ell_j$ with values $x_1=t, x_2=\frac{1}{2^{c_2}}, x_3=x_4=0$.

Assume that the time spent at $\ell_i$, $t=1-\frac{1}{2^{c_1-1}}+\epsilon$. 
Now consider the case of going to the location $C$ from $B$ resetting $x_2$.
At $C$, we have $x_1=\frac{1}{2^{c_1}}+t=1-\frac{1}{2^{c_1}}+\epsilon$, 
$x_2=0, x_3=1-\frac{1}{2^{c_1-1}}+\epsilon, x_4=0$.
The gadget $GetProb$ is entered with values $x_1=0, x_2=\frac{1}{2^{c_1}}-\epsilon, 
x_3=1-\frac{1}{2^{c_1}}, x_4=0$. The initial location of $GetProb$ is $E0$.

 A total of 2 units of time can be spent at $E0$. It can be seen that 
  the time intervals $[0, \frac{1}{2^{c_1}}]$ and $[1,2-\frac{1}{2^{c_1}}+\epsilon]$ respectively 
     are enabled to reach $E3$ and $E4$. Similarly, the time intervals 
     $[\frac{1}{2^{c_1}},1]$ and $[2-\frac{1}{2^{c_1}}+\epsilon,2]$ respectively are enabled
  to reach $E1$ and $E2$. The probabiltiy 
  of reaching  $E3$ or $E4$ is thus $\frac{1}{2}(1+\epsilon)$ and 
  the probability of reaching $E1$ or $E2$ is thus 
  $\frac{1}{2}(1-\epsilon)$. The locations $P1, P2$ are reached with 
  $x_1=0, x_2=\frac{1}{2^{c_1}}-\epsilon, x_3=1-\frac{1}{2^{c_1}}, x_4=0$.   
  The probabilty of reaching a target location through $P1$ (from $E0$) is 
  hence $\frac{1}{2}(1+\epsilon)\frac{1}{2}(1-\epsilon)=\frac{1}{4}(1-\epsilon^2)$, while 
  the probability of reaching a target location through $P2$ (from $E0$) is $\frac{1}{2}(1+\epsilon)\frac{1}{2}(1-\epsilon)=\frac{1}{4}(1-\epsilon^2)$. The probability of reaching a target location in $GetProb$ is thus, $\frac{1}{2}(1-2\epsilon^2)$. 
  Note that if we start with $t=1-\frac{1}{2^{c_1-1}}-\epsilon$, we obtain exactly the same probability. 
  Thus, the probabilty to reach a target location in $GetProb$ is $\frac{1}{2}$ iff $\epsilon=0$. 

\section{Time-bounded quantitative reachability for $2\frac{1}{2}$ STGs}
\label{sec:time-bounded-app}
The details of the zero check (and the proof that it can be done in bounded time), which were missing in the main paper, due to lack of space, are given below.
\input{zero-check.tex}

\noindent{\bf Halting and Correctness of construction}
 \input{halt-final.tex}

\section{Details for Section \ref{sec:quant-dec}}
\label{app:quant-dec}

\subsection{Timed Region Graph}
We begin with a formal definition of the timed region graph. 
Given a $1\frac{1}{2}$ STG $\Gg=(\mathcal{A},L, \omega, \mu)$, we define the timed region graph 
$\GR=(\mathcal{R}(\mathcal{A}), L \times \mathcal{R}, \omega^{\mathcal{R}}, \mu^{\mathcal{R}})$ where 
$\mathcal{R}(\mathcal{A})$ has as its locations ordered pairs $(\ell, R)$ where
$\ell \in L$  and $R$ is a classical region. 
The transitions of $\GR$ are defined as follows.
We have a transition $(\ell, R) \stackrel{guard(R''),e,Y}{\rightarrow}(\ell', R')$
iff there exists an edge $e=\ell \stackrel{g,Y}{\rightarrow}\ell'$ in $\mathcal{A}$ such that there exists $\nu \in R, t \in \mathbb{R}$ with 
$(\ell,\nu)\stackrel{t,e}{\rightarrow}(\ell', \nu+t)$, $\nu+t \in R''$, 
and $\nu'=\nu''[Y \leftarrow 0] \in R'$. Here, $guard(R'')$ represents the minimal guard that captures region $R''$.
For instance, if region $R''$ is $(0,1)$ then $guard(R'')$ is $0 < x < 1$. 
Also, $Y$ is either the emptyset, or the single clock $\{x\}$.  
The standard region automaton (Alur-Dill) can be recovered by labelling transitions of $\GR$ 
with only $e$ rather than with  $guard(R''), e, Y$.

For every state $s=(\ell, \nu)$ in $\mathcal{A}$, there is a mapping $\imath(s)$
which maps it to $(\ell, R)$ such that $\nu \in R$. The 
probability measure for $\GR$ is defined such that $\mu^{\mathcal{R}}_{\imath(s)}=\mu_s$ and 
the weights of edges are also preserved. That is 
$\omega^{\mathcal{R}}(f)=\omega(e)$ where $f=guard(R''),e,Y$ is the edge corresponding to $e$, obtained 
from the map between states.  For brevity, we decorate the transitions 
in Figure \ref{ex1} with only $e_i$ rather than 
$guard(R''), e_i, Y$.

A strategy $\sigma$ of $\Diamond$ in $\Gg$ is a function that maps a finite run
$\rho= (l_0, \nu_0) \xrightarrow{d_0,e_0} (l_1, \nu_1)
\xrightarrow{d_1,e_1} \dots (l_n, \nu_n)$ to a transition $(d,e)$ where $d \in \mathbb{R}^+$ and $e$ is an edge, such
that $(l_n, \nu_n) \xrightarrow{d,e} (l', \nu')$ for some $(l',
\nu')$, whenever $l_n \in L_{\Diamond}$. For each such strategy $\sigma$ in $\Gg$, 
we have a corresponding strategy $\imath(\sigma)$ in $\GR$
that maps the finite run $\imath(\rho)=(l_0, R_0) \xrightarrow{f_0} (l_1, R_1)
\xrightarrow{f_1} \dots (l_n, R_n)$ to a transition $f$ such
that $(l_n, R_n) \xrightarrow{f} (l', R')$ for some $(l',
R')$, whenever $l_n \in L_{\Diamond}$. Here, $f_i$ stands for 
$guard(R''_i), e_i, Y$ such that $\nu_i+d_i \in guard(R''_i)$. 
Moreover,  $(l_i, R_i)=\imath(l_i, \nu_i)$ and 
 $\nu_i \in R_i$
for all $i$.  
For every finite path $\pi((l, \nu),e_1 \dots e_n)$ in $\Gg$, we have a finite set  
of paths $\pi( ((l,R), \nu), f_1 \dots f_n)$ in $\GR$, each one corresponding to a choice of regions passed. 
If $\rho$ is a run in $\Gg$, $\imath(\rho)$ stands for the unique image of the run in $\GR$.

\begin{lemma}[Strategy Mapping between $\Gg$ and $\GR$]
Let $\Gg$ be a $1 \frac{1}{2}$ player STG. Then 
player $\Diamond$ has a strategy $\sigma$ in $\Gg$ to reach  $(l, \nu)$ in $\Gg$ with probability 
$\sim c$ iff $\Diamond$ has a strategy $\imath(\sigma)$ in $\GR$ 
to reach $\imath(l, \nu)$ with the same probability. 
\end{lemma}
\begin{proof}
The proof follows by construction of $\GR$ from $\Gg$. 
Fix a strategy $\sigma$ in $\Gg$. At each $(l, \nu)$ such that $l \in L_{\Diamond}$,
$\sigma$ chooses a time delay $d$ and an edge $e$ from $(l, \nu)$ 
based on the path $\rho$ seen so far, such that 
$(l, \nu)$ is the last state in $\rho$.  
Let $\rho=(l_0, \nu_0={\bf{0}}) \stackrel{d_0,e_0}{\rightarrow} (l_1, \nu_1) \stackrel{d_1,e_1}{\rightarrow} \dots \stackrel{d_{n-1},e_{n-1}}{\rightarrow} (l_n, \nu_n)=(l, \nu)$.

We induct on the number of stochastic nodes seen so far in $\rho$. Assume that in the path so far, we have witnessed exactly one stochastic node. 

\begin{enumerate}
\item Assume $|\rho|=1$ and $l_0$ is a stochastic node. 
In $\GR$, we start with $(l_0,R_0)$ where $R_0={\bf 0}$ is the initial region.  
 To satisfy the guard on edge $e_0$ in $\Gg$, we can choose any appropriate delay 
  $d_0$.  In $\GR$, 
 the guard chosen is the minimal region which contains $d_0$. 
  For each choice of $d_0$, we have an appropriate guard 
 which captures the correct interval which contains it.  This time interval determines the probability 
 for the edge $e_0$ chosen in both $\Gg$ as well as $\GR$ and is the same, by setting the limits of the integral. 
 
 If $l_0$ is not a stochastic node, then we simply continue 
 mapping locations in $\Gg$ with those in $\GR$, by mapping edges $e_i$ with $f_i$, until we reach a stochastic node. 
  The first time we reach a stochastic node with valuation $\nu_i$, $(l_i, \nu_i)$ in $\rho$, we will reach 
 in $\GR$, the node $(l_i, R_i)$. At this point, as seen above, 
 for a delay $d_i$ and an edge $e_i$ chosen in $\Gg$, we choose $f_i$ so that the minimal guard 
 captures the precise time interval in which $d_i+\nu_i$ lies in. Since the minimal time interval containing $d_i+\nu_i$ determines the probability 
 of $e_i$ in $\Gg$ and $f_i$ in $\GR$, we have matched the probabilities 
 till the first stochastic node.  
 
 \item Now assume that the probabilties are preserved till some $n-1$ stochastic nodes seen, and we are going to see 
 the $n$th stochastic node.  The same argument as above applied to the $n$th stochastic node 
  ensures that the probabilities incurred each time remain the same, and hence 
  the probability of reaching some $(l, \nu)$ in $\Gg$ is same as that of reaching $\imath(l, \nu)$
  in $\GR$.
 
\end{enumerate}

\end{proof}

\begin{lemma}
\label{lem:quant}
If $\Gg$ is a initialized 1 clock $1 \frac{1}{2}$ player STG, then $M_{\Gg}$ is a Markov decision process.
\end{lemma}
\begin{proof}[Proof sketch]
Observe that since $\Diamond$ to $\bigcirc$  edges always reset the clock, we can compute the probability values 
 on $\bigcirc$ nodes. We need to show that from any $\bigcirc$ node,
 the probability of the outgoing paths (and edges to $\Diamond,
 \bigcirc$ nodes) adds up to 1.  

First observe that if $N=0$, i.e, if $\GR$ has no stochastic nodes $(l,\alpha)$ s.t. $\alpha\not \in \{0,\infty\}$, then $\GR$ already defines an MDP, obtained by computing the discrete probability on the edges (follows from the definition of an initialized STG: the absence of zero and unbounded regions in the stochastic nodes implies the absence of cycles in the STG).

Then, we recursively, remove all deletable nodes $\GR$ to obtain new region graph STG $G$ (with a new path-labeling alphabet on its edges), where the probabilities of any paths between nodes of $G$ are the same as the probability of that path in $\GR$. Thus, the sum of all probabilities of outgoing paths add to 1. Now, as all deletable nodes are removed, this gives an MDP. 

We now elaborate on the construction of $M_{\Gg}$ given the STG $\Gg$. Let $\Gg=(\mathcal{A},(L_{\Diamond},L_{\bigcirc}),\omega,\mu)$ be an 1 clock $1 \frac{1}{2}$ player STG.  Let us look at the region graph $\GR$ corresponding to it. 
 Further let $(l,R)$ be a deletable node in $\GR$. Then we define $remove(l,R)$ which modifies the region graph $\GR$ by 
\begin{itemize}
\item removing this node and all edges incoming to and outgoing from this node.
\item for each incoming edge $e_1$ from, say, $(l_1,R_1)$ to $(l,R)$ and each outgoing edge $e_2$ from $(l,R)$ to, say $(l_2,R_2)$, 
we add a new direct edge from $(l_1,R_1)$ to $(l_2,R_2)$ with the new label $e_1 e_2$.
\end{itemize}
Note that this operation is well-defined since, for every deletable node, there must exist an incoming edge (since the region is non-zero). Further,  there must also exist an outgoing edge, since it is a stochastic node and hence the sum of probabilities on outgoing edges of stochastic nodes in $\Gg$ sums to 1. If there is a self-loop, then it must be reset (by the structural non-Zeno assumption) and then this node will not be deletable. 

Let $G_1$ be the resulting structure obtained after the remove operation. The probability of these new edges labeled by paths in $G_1$ is the probability of the respective paths in $\Gg$.

\begin{lemma}
Suppose $G_1$ is obtained from $\GR$ by performing $remove (l,R)$. Then, for each stochastic node in $G_1$ the sum of outgoing probabilities is 1.
\end{lemma}
\begin{proof}
Consider any node $(l',R')$ in $G_1$ such that $l'\in L_\bigcirc$. There are two cases:
\begin{itemize}
\item there is no edge in $\GR$ from $(l',R')$ to $(l,R)$. Then the outgoing probabilities of $(l',R')$ do not change in $G_1$. As they summed to 1 in $\Gg$, they will continue to do so in $G_1$.
\item there is an edge $e$ in $\Gg$ from $(l',R')$ to $(l,R)$. Then consider all outgoing edges from $(l,R)$ in $\Gg$, call them $e_1,\ldots e_k$. By stochasticity of $\Gg$, $\sum_{i=1}^k\mathcal{P} (\pi((l,R),e_i))=1$. Then in $G_1$ from $(l',R')$, we have exactly $k$ outgoing edges labeled $ee_1,ee_2,\ldots e e_k$. Now if $E$ is the set of all other ($\neq e$) edges outgoing from $(l',R')$, then the sum of probabilities of all outgoing edges from $(l',R')$ is given by $\sum_{i=1}^k\mathcal{P} (\pi((l',R'), ee_i))+\sum_{e'\in E} \mathcal{P}(\pi((l',R'), e'))$ which is 
\begin{align*}
&=\mathcal{P}(\pi((l',R'),e))\cdot \sum_{i=1}^k\mathcal{P}(\pi((l,R),e_i))+ \sum_{e'\in E} \mathcal{P}(\pi((l',R'),e'))\\
&=\mathcal{P}(\pi((l',R'),e))+\sum_{e'\in E} \mathcal{P}(\pi((l',R'),e'))=1
\end{align*}
This follows by linearity of the Lebesgue integral and stochasticity of $\Gg$.
\end{itemize}
\end{proof}
Thus, $G_1$ is an (extended) STG in which edges are labeled by paths instead of edges and the probability of paths are computed as before. Thus by now repeatedly applying the remove operation on all deletable edges we obtain (after finitely many steps) an (extended) STG $G_n$ in which there are no deletable edges. This implies that $G_n$ is an MDP. Note that as an immediate consequence of the above lemma we also obtain that the probability of all paths are preserved. 
\end{proof}

\begin{lemma}[Strategy Mapping between $\GR$ and $M_{\Gg}$]
Let $\GR$ be the timed region graph corersponding to a  $1 \frac{1}{2}$ player STG $\Gg$. Then 
player $\Diamond$ has a strategy $\sigma$ in $\GR$ to reach  $(l,R)$ in $\GR$ with probability 
$\sim c$ iff $\Diamond$ has a strategy $g(\sigma)$ in $M_{\Gg}$ 
to reach $(l,R)$ with the same probability. 
\end{lemma}
\begin{proof}
There are two parts to the proof.
\begin{enumerate}
\item[(a)] Let $(l, R)$ and $(l', R')$ be two nodes in $M_{\Gg}$. Then for every path 
$\pi$ between $(l, R)$ and $(l', R')$ in $\GR$, we have a path $\pi'$  
in $M_{\Gg}$ and conversely. The probabilities of $\pi, \pi'$ are same 
in $\GR$ and $M_{\Gg}$. 
\item[(b)]  Show that for every strategy $\sigma$ in $\GR$, there exists a strategy $g(\sigma)$ in $M_{\Gg}$ 
that preserves probabilities.
\end{enumerate}
We can prove (a) and (b) together. 
 Consider $(l, R)$ and $(l', R')$ in $M_{\Gg}$ such that 
$l \in L_{\Diamond}$. Let $(l, R) \stackrel{f_1}{\rightarrow} (l_1, R_1) \stackrel{f_2}{\rightarrow} \dots \stackrel{f_n}{\rightarrow} (l_n, R_n) \stackrel{f}{\rightarrow}
 (l', R')$ be a path $\pi$ in $\GR$, according to a strategy $\sigma$ in $\GR$.  Lets see what happens to this path 
 in $M_{\Gg}$.  The operation of $remove(l_i, R_i)$  might remove some of the 
 intermediate nodes of $\pi$ (excluding the first and last, since by assumption they are in $M_{\Gg}$). 
 Let $(l_i, R_i)$ be the first such node to be deleted. Then in $M_{\Gg}$, we have all the nodes 
 from $(l, R)$ to $(l_{i-1}, R_{i-1})$. According to strategy $\sigma$, $(l_j, R_j)$ 
 has been selected based on the prefix till $(l_{j-1}, R_{j-1})$ whenever $l_{j-1} \in L_{\Diamond}$. 
 Clearly, in $M_{\Gg}$, if all nodes until $(l_i, R_i)$ are carried forward, 
 then the strategy chosen at all nodes $(l_j, R_j)$, $l_j \in  L_{\Diamond}$, 
 $j < i$ is the same as $\sigma$.
 
  If $(l_i, R_i)$ is deleted, clearly, $l_i \in L_{\bigcirc}$, 
 and $R_i \neq {\bf 0, \infty}$. Then, $l_{i-1} \notin L_{\Diamond}$, by definition 
 of initialized STG. Let   $(l_k, R_k)$, $k < i$ be the last node from 
 $L_{\Diamond}$ before $(l_i, R_i)$. By the delete operation, we obtain the path 
 $(l, R) \stackrel{f_1}{\rightarrow} (l_1, R_1) \stackrel{f_2}{\rightarrow} \dots (l_k, R_k) \stackrel{f_{k+1}}{\rightarrow} \dots 
 (l_{i-1}, R_{i-1}) \stackrel{f_if_{i+1}}{\rightarrow} 
 (l_{i+1}, R_{i+1})$  till $(l_{i+1}, R_{i+1})$ in $M_{\Gg}$. Continuing this, when we finish removing all deletable nodes, we 
 obtain the path $\pi'$ in $M_{\Gg}$ such that if nodes $(l_i, R_i),(l_{i+1}, R_{i+1}), \dots, 
 (l_s, R_s)$ are deleted, then we obtain the edge $(l_{i-1}, R_{i-1}) \stackrel{f_if_{i+1} \dots f_{s+1}}{\rightarrow} (l_{s+1}, R_{s+1})$ in $M_{\Gg}$. 
 For any path $\pi$ in $\Gg$, we obtain a unique path $\pi'$ in $M_{\Gg}$. 
 The strategy    $g(\sigma)$ in $M_{\Gg}$ is defined from strategy $\sigma$ in $\Gg$ as follows: 
   
  \begin{itemize}
 \item 
   If $\sigma$ maps $(l_h, R_h)$ to $(l'_h,R'_h)$ choosing edge $f$ based  
  on a path $\pi$ such that  $(l_h, R_h)$  is the last node of $\pi$, then in $M_{\Gg}$, 
  $g(\sigma)$ maps $(l_h, R_h)$ to $(l'_h,R'_h)$ choosing edge $f$ based  
  on the unique path $\pi'$ corresponding to $\pi$.  Note here that the only change 
  in the strategy $g(\sigma)$ as compared to $\sigma$ is the path $\pi'$ seen so far, obtained by deleting some nodes from 
  $\pi$.  
  \end{itemize}
 
 Given a path $\pi$ in $\GR$ as above, the probability 
 of the path is obtained from the edges $f_1f_2 \dots f_n f$. Since the sequence of labels 
 on the path $\pi'$ are exactly same as $f_1f_2 \dots f_n f$, the probability of $\pi$ and 
 $\pi'$ are the same. Since this is true about all paths $\pi$ in $\GR$, we have 
 the probability of reaching $(l', R')$ from $(l, R)$ 
 in $\Gg$ is same as the probability of reaching  
 $(l', R')$ from $(l, R)$ in $M_{\Gg}$, for any two nodes $(l', R')$, $(l, R)$
 in $M_{\Gg}$. 
 \end{proof}

\section{Example of a 2-clock STA with unfair runs}
\label{app:example}

This example has been taken from \cite{journal-sta} to help the reader get an intuition 
of why two clocks or the uninitialized condition  creates problems even in qualitative reachability. 
Our assumptions of 1-clock and initialized-ness circumvent these problems even for quantitative reachability.

  \begin{figure}[h]
   \scalebox{.8}{
   \begin{tikzpicture}[->,thick]                                                                                                                                          
    \node[cir,label=below:$y < 1$] at (-14,0) (A1) {E} ;   
     \node[cir] at (-11,0) (B) {C}; 
        \node[cir,label=below:$y \leq 2$,initial, initial text ={}, initial where=above] at (-7,0) (C) {B};  
   \node[cir] at (-4,0) (D) {F};
    \node[cir] at (-2,0) (E) {G};
    \path (B) edge node [above]{$e_5, y=2$}(A1);
    \path (B) edge node [below]{$y:=0$}(A1);
\path (C) edge node [above]{$e_1, y=2$}(D);
    
    \path (D) edge node [above]{$e_2, y=1$}(E);
    \path (D) edge node [below]{$y:=0$}(E);
    \path (E) edge[bend left] node [below]{$e_3, x>1, x:=0$}(C);
    
    \path (C) edge node [above]{$e_4, 1 < y <  2$}(B);
        \path (A1) edge[bend right] node [above]{$e_6, x>2$}(C);
        \path (A1) edge[bend right] node [below]{$x:=0$}(C);

 \end{tikzpicture}                                                                                                                                                      
}
 \end{figure}

In this example, one does not reach location $G$ almost surely, even though 
thats what one would conclude by working on the region graph. 
Every fair run using edges of non-zero probability indeed visits $G$ infinitely often. 
However, the problem is that the run $(e_4e_5e_6)^{\omega}$ has a non-zero probability. Thus, 
there is an unfair run in the automaton with a non-zero probability, and hence 
one cannot reach $G$ almost surely.

The interplay of the clocks $x,y$ is very useful here. In fact, if one starts in  node $B$ 
with $x=0,y=t_0$, then one reaches $E$ with $(2-t_0, 0)$. The enabled interval for edge 
$e_4$ is $(1-t_0, 2-t_0)$, while that for $e_6$ is $t_1 \in (t_0,1)$. Again, 
$e_4$ is enabled with time interval $(1-t_1, 2-t_1)$,
 while 
$e_6$ is enabled with $t_2 \in (t_1, 1)$ and so on. 

\begin{align*}
\mathcal{P}(\pi((B,(0,t_0)), (e_4e_5e_6))) &=\frac{1}{2-t_0}\int_{t=1-t_0}^{2-t_0}\frac{1}{1-t_0} \int_{t_1=t_0}^{1} dt_1 dt\\
&=\frac{1}{2-t_0}.\frac{1}{1-t_0} \int_{t_1=t_0}^{1}dt_1dt 
\end{align*}

In particular, it can be shown that 
\begin{align*}
\mathcal{P}(\pi((B,(0,t_0)), (e_4e_5e_6)^n)) &=\frac{1}{2-t_0}\int_{t=1-t_0}^{2-t_0}\frac{1}{1-t_0} \int_{t_1=t_0}^{1} \mathcal{P}(\pi((B,(0,t_1)), (e_4e_5e_6)^{n-1}))dt_1 dt\\
&=\frac{1}{2-t_0}.\frac{1}{1-t_0} \int_{t_1=t_0}^{1}\mathcal{P}(\pi((B,(0,t_1)), (e_4e_5e_6)^{n-1}))dt_1 dt
\end{align*}

By an inductive argument, \cite{journal-sta} shows that 
$\mathcal{P}(\pi((B,(0,t_0)), (e_4e_5e_6)^n))=\frac{t_0}{2-t_0}>0$, and 
$\mathcal{P}(\pi((B,(0,t_0)), (e_4e_5e_6)^{\omega}))>0$. 

Note that this example is an uninitialized STA with 2 clocks. If  one makes this example initialized, by resetting both $x,y$ on a transition 
(on $e_3, e_6$), then again it can be seen that the resulting automaton (Figures 10,11 in \cite{journal-sta}) also has unfair runs of non-zero probability.

%% file: zero-check.tex
Let us consider (wlog) the case when the $(k+1)^{th}$ instruction checks whether counter $C_1$ is zero. Assume that after $k$ instructions, we have $x_1=\frac{1}{2^{k+c_1}3^{k+c_2}}, x_2=0, z=1-\frac{1}{2^k}$ and $a=b=0$. The main module, given in Figure \ref{zer-chk}, can be divided into two parts. 
\begin{figure} [!h]                                                                                                                                                          
   \begin{center}                                                                                                                                                         
   \scalebox{0.9}{                                                                                                                                                       
   \begin{tikzpicture}[->,thick]                                                                                                                                          
    \node[initial, dia,initial text={$a,b,x_2=0$}] at (0,-1) (A) 
    {$\ell_{k+1}$} ;  
       \node[dia] at (2,-1) (B) {B};
    \node[box,label=below:{$b= 0$}] at (4,-1) (C) {Check};
       \node[dia,label=right:{$b= 0$}] at (7,-1) (D) {$D$};
    \node[box,label=above:{$b= 0$}] at (8,1) (D1) {$=0$};
    \node[box,label=below:{$b= 0$}] at (8,-3) (D2) {$>0$};
    \node[dia] at (12,1) (E1) {$\ell_{k+2}$};
    \node[dia] at (12,-3) (E2) {$\ell'_{k+2}$};
    
       \node[rounded rectangle,fill=gray!20!white] at (11,-.5) (W00) {Rem$^{x_1}_k$};
     \node[rounded rectangle,fill=gray!20!white] at (11,-1.5) (W01) {Rem$^{x_1}_k$};
   
       \node[rounded rectangle,fill=gray!20!white] at (13.5,-.5) (W1) {Wid$_{=0}$};
   \node[rounded rectangle,fill=gray!20!white] at (13.5,-1.5) (W2) {Wid$_{>0}$};

             \node[rounded rectangle,fill=gray!20!white] at (2,-3) (chkz) {Check $z$};
         \node[rounded rectangle,fill=gray!20!white] at (6,-3) (chkx) {Check $x$};
    \path (A) edge node [midway,above] {$a<1$}(B);      
    \path (A) edge node[midway,below] {$x_2:=0$}(B);
    \path (B) edge node[midway,above] {$a<1$}(C);
    \path (B) edge node[midway,below] {$b:=0$}(C);
    \path (C) edge node[] {}(chkz);
    \path (C) edge node[midway,below] {$x_1:=0$}(D);
       \path (C) edge node[] {}(chkx);
    \path (D) edge node[] {}(D1);
    \path (D) edge node[] {}(D2);
    \path (D1) edge node[midway,above] {$a:=0$}(E1);
    \path (D2) edge node[midway,below] {$a:=0$}(E2);
    \path (D1) edge node[midway,right] {$z:=0$}(W00);
    \path (D2) edge node[midway,right] {$z:=0$}(W01);
    \path (W00) edge node[midway,above] {$a=1$}(W1);
    \path (W00) edge node[midway,below] {$a:=0$}(W1);
    
    \path (W01) edge node[midway,above] {$a=1$}(W2);
    \path (W01) edge node[midway,below] {$a:=0$}(W2);

   \end{tikzpicture}                                                                                                                                                      
  }                                                                                                                                                                     
  \caption{Zero Check $C_1(x_1)$. $x_1$ holds the value $\frac{1}{2^{c_1+k}3^{c_2+k}}$ on entering the module.}                                                                                                                             
   \label{zer-chk}                                                                                                                                                        
  \end{center}                                                                                                                                                           
\end{figure}
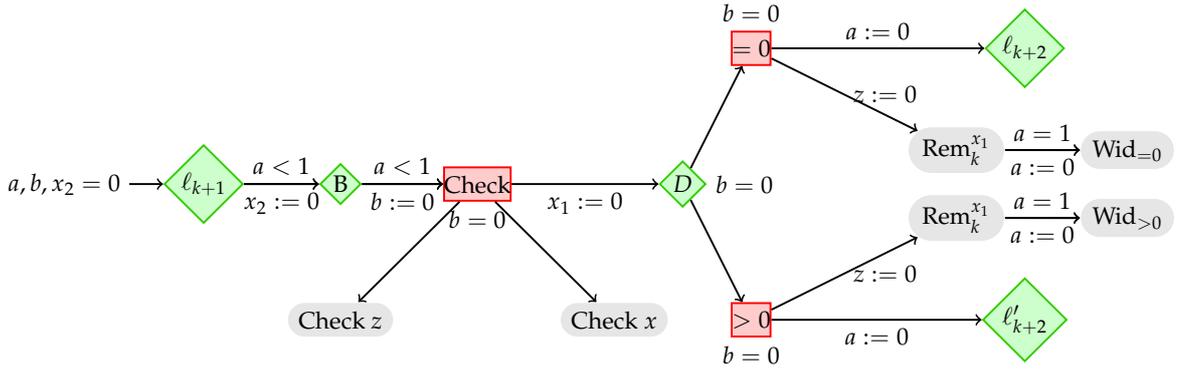

\begin{enumerate}
\item First, we make sure that the instruction counter, i.e., Clock $z$ is updated correctly: we spend times $t_1,t_2$ at locations $\ell_{k+1}, B$ respectively, and check that $t_1+t_2=\frac{1}{2^{k+1}}$ and $t_2=\frac{1}{2^{k+1+c_1}3^{k+1+c_2}}$. For this it suffices to check that at location Check we have $x_2=\frac{1}{6}.\frac{1}{2^{k+c_1}3^{k+c_2}}$, $x_1=\frac{1}{2^{k+c_1}3^{k+c_2}}+t_1+t_2, z=1-\frac{1}{2^{k+1}}$ $a=t_1+t_2$ and $b=0$. This is done, as before, by the Player $\Box$ using widgets Check $z$ (given in Figure~\ref{check1}) and Check $x$ similar to the widget Check $x_2$ in Figure \ref{check1}, where one simply changes the weights on edges of $F1$ to $C1$ and $C2$ to 
1 and 5 respectively. Then, we proceed to $D$. 

\item At $D$, player $\Diamond$ guesses whether $C_1=0$ or not, by choosing an appropriate $\Box$ location. From these, player $\Box$ can either allow the simulation to continue, or check the correctness of $\Diamond$'s guess. This check is done in three steps:
\begin{enumerate} 
\item First, we eliminate $k$ from $\frac{1}{2^{c_1+k}3^{c_2+k}}$ by multiplying by 6 for $k$ times, and from $a=\frac{1}{2^k}$ obtaining $a=1$. Each time multiplication by 6 happens, the clocks $x_1, x_2$ alternate. The widgets Rem$_k^{x_1}$ and   Rem$_k^{x_2}$ (Figure \ref{wid-chk}) are used alternately as long as $a<1$, and $x_1, x_2$ alternately store values $\frac{1}{2^{c_1+k}3^{c_2+k}}$, $\frac{1}{2^{c_1+k-1}3^{c_2+k-1}}$ till $\frac{1}{2^{c_1}3^{c_2}}$ is obtained in one of $x_1, x_2$. 
\item Once  $\frac{1}{2^{c_1}3^{c_2}}$ is obtained in $x_1$ or $x_2$, we further multiply by 3 for $c_2$ times to obtain $\frac{1}{2^{c_1}}$. This is done as represented in widgets Wid$_{=0}$, Wid$_{>0}$. 
\item Finally, to check if player $\Diamond$'s guess is correct or not, we only need to check if $x_1$ or $x_2$ is 1 which corresponds to $c_1=0$. 
\end{enumerate}
\end{enumerate}
It can be seen that a target location is reached with probability $\frac{1}{2}$ from Figure \ref{zer-chk} iff 
(1) the $(k+1)$th instruction (zero check) is accounted for correctly, at locations $\ell_{k+1}$ and $B$ 
in figure \ref{zer-chk}. The widgets Check $z$ and Check $x$ check this. (2) Player $\Diamond$ guesses correctly
whether $C_1$ is zero or not. If player $\Box$ goes in for further checks, then 
player $\Diamond$ must be faithful in the widgets  Rem$^{x_1}_k$ and  Rem$^{x_2}_k$, and 
also in widgets Wid$_{=0}$ and Wid$_{>0}$.

\begin{figure} [!h]                                                                                                                                                          
   \begin{center}                                                                                                                                                         
   \scalebox{0.9}{                                                                                                                                                       
   \begin{tikzpicture}[->,thick]       
                                                                                                                                      
    \node[initial, dia,label=above:{$z= 0$},initial text={$a=\frac{1}{2^{k+1}},x_2=\frac{1}{2^{k+c_1+1} 3^{k+c_2+1}},z,b,x_1=0$}] at (0,-1) (A) 
    {$A2$} ;
     \node[dia] at (2,-1) (A0) 
    {$A0$} ;  
    \node[cir,label=below:{$z= 0$}] at (0,-3) (V0) 
    {T} ; 
    \node[cir,accepting] at (-2,-3) (V1) 
    {} ; 
    
    \node[cir] at (2,-3) (V2) 
    {} ;

          \node[dia] at (4,-1) (B) {B2};
    \node[box,label=above:{$b= 0$}] at (6,-1) (C) {Check};
    \node[rounded rectangle,fill=gray!20!white] at (6,-3) (C0) {Wid$_{=0}$ OR Wid$_{>0}$};
                     \node[rounded rectangle,fill=gray!20!white] at (9,-1) (chk) {Rem$^{x_2}_{k}$};
                 
                 \node[rounded rectangle,fill=gray!20!white] at (3,-3) (mula) {Mul $a$};
         \node[rounded rectangle,fill=gray!20!white] at (9,-3) (mulx) {Mul $x_2$};

    \path (A) edge node [midway,above] {$a < \frac{1}{2}$}(A0);
    \path (A) edge node [midway,right] {$a=\frac{1}{2},x_2=\frac{1}{6}$}(V0);
            
    \path (A0) edge node [midway,above] {$x_1\leq 1$}(B);      
    \path (A0) edge node[midway,below] {$x_1:=0$}(B);
    \path (B) edge node[midway,above] {$x_1 \leq 1$}(C);
    \path (B) edge node[midway,below] {$b:=0$}(C);
        \path (C) edge node[midway,above] {$a<1$}(chk);
        \path (C) edge node[midway,below] {$z,x_2:=0$}(chk);
        \path (C) edge node[midway,left] {$a=1$}(C0);
        \path (C) edge node[midway,right] {$a:=0$}(C0);
        \path (V0) edge node[midway,right] {}(V1);
        \path (V0) edge node[midway,right] {}(V2);
        
        \path (C) edge node[] {}(mulx);
    \path (C) edge node[] {}(mula);
        \end{tikzpicture}                                                                                                                                                      
  }                                                                                                                                                                     
  \caption{Rem$^{x_1}_{k}$: Times $t_1,t_2$ spent at $A0,B2$ such that $t_1+t_2=\frac{1}{2^k}$ and $t_2=\frac{1}{2^{c_1+k}3^{c_2+k}}$.
  Note that $k \geq 2$. Mul $a$ checks on $t_1+t_2$ while Mul $x_2$ checks on $t_2$.
  Note that if Rem$^{x_1}_k$ is entered from the $>0$ $\Box$ location of Figure \ref{zer-chk}, then the target in Rem$^{x_1}_k$ is not reached with probability $\frac{1}{2}$  when $a=\frac{1}{2}$ and $x_2 = \frac{1}{6}$ as this corresponds to the scenario where  $C_1 = C_2 =0$ implying an incorrect guess by Player $\Diamond$ that $C_1 >0$.}                                                                                                                    
   \label{wid-chk}                                                                                                                                                        
  \end{center}                                                                                                                                                           
\end{figure}
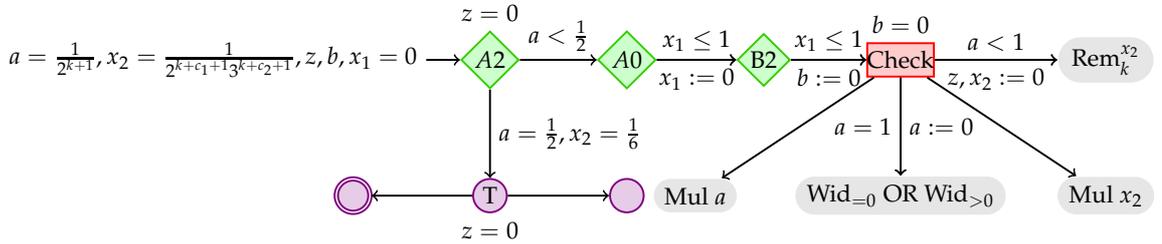

\begin{figure} [!h]                                                                                                                                                          
   \begin{center}                                                                                                                                                         
   \scalebox{0.9}{                                                                                                                                                       
   \begin{tikzpicture}[->,thick]                                                                                                                                          
    \node[initial, cir,label=right:{$b=0$},initial text={}] at (-4,-1) (A) 
    {$A5$} ;  
       \node[cir,label=below:{$b\leq 1$}] at (-2,0) (B) {B3};
       \node[cir,accepting] at (0,0) (G) {};
       \node[cir] at (-3,1) (H) {};

       \node[dia] at (-3,-2) (C) {C3};
       \node[dia] at (-1,-2) (D) {D3};
       \node[cir,label=above:{$b\leq 1$}] at (1,-2) (E) {E3};
       \node[cir,accepting] at (2.5,-2) (I) {};
       \node[cir] at (1,-4) (J) {};
           \path (A) edge node [midway,above] {}(B);
           \path (A) edge node [midway,above] {}(C);
            \path (B) edge node [midway,above] {$z > 1$}(G);
            \path (B) edge node [midway,left] {$z \leq 1$}(H);
            \path (C) edge node[midway,above] {$z=1$}(D);
            \path (C) edge node[midway,below] {$z:=0$}(D);
            \path (D) edge node[midway,above] {$a=2$}(E);
            \path (D) edge node[midway,below] {$b:=0$}(E);
            \path (E) edge node[midway,above] {$z >1$}(I);
            \path (E) edge node[midway,left] {$z \leq 1$}(J);

           \node[initial, cir,label=above:{$b\equal 0$},initial text={}] at (5,-1) (A1) {A4} ;
    \node[dia] at (6,-1) (B1) {B4};
             \node[cir,label=below:{$z\leq 1$}] at (8,-1) (B11) {C4};
    \path (A1) edge node[midway,above] {}(B1);
    \node[cir,accepting] at (9.5,-1) (B3) {};
    \node[cir] at (8,0.5) (B4) {};
\path (B1) edge node[midway,above] {$z=1$}(B11);
\path (B1) edge node[midway,below] {$z:=0$}(B11);
\path (B11) edge node[midway,above] {$x_2\leq 2$}(B3);
    \path (B11) edge node[midway,left] {$x_2>2$}(B4);
\node[cir,label=left:{$b\leq 1$}] at (5,-3) (C4) {E4};
   \node[cir,accepting] at (7,-3) (C5) {};
\node[cir,label=left:{$b=0$}] at (5,-2) (H) {};
\node[dia] at (7,-2) (H1) {};
\node[dia] at (5,-4.5) (C6) {};
\path (H) edge node[midway,above] {5}(H1);
\path (C4) edge node[midway,above] {$x_1 \geq 1$}(C5);
\path (C4) edge node[midway,left] {$x_1<1$}(C6);
 \path (A1) edge node[midway,above] {}(H);
 \path (H) edge node[midway,left] {1}(C4);

        \end{tikzpicture}                                                                                                                                                      
  }                                                                                                                                                                     
  \caption{Mul $a$ and Mul $x_2$. On entry, $x_1=t_2, b=0, x_2=n+t_1+t_2, z=t_1+t_2, a=\frac{1}{2^{k+1}}+t_1+t_2$ for $n=\frac{1}{2^{k+c_1+1} 3^{k+c_2+1}}$.
  Probabilty to reach a target in Mul $a$ is $\frac{1}{2}(t_1+t_2)+\frac{1}{2}(1-\frac{1}{2^{k+1}})$, while that in Mul $x_2$ is 
  $\frac{1}{2}(1-n)+\frac{1}{2}\frac{1}{6}t_2$. Thus, a target is reached in Mul $a$ with probability $\frac{1}{2}$ iff $t_1+t_2=\frac{1}{2^{k+1}}$. 
  That makes $a=\frac{1}{2^k}$ at the end. Likewise, to get a probability $\frac{1}{2}$ in Mul $x_2$, we need 
  $t_2=6n$.
  }                                                                                                                             
   \label{mul-a}                                                                                                                                                        
  \end{center}                                                                                                                                                           
\end{figure}
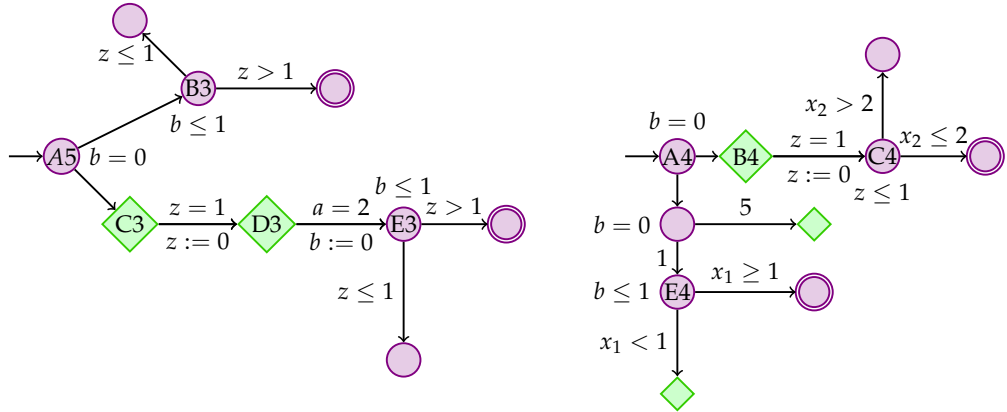

\begin{figure} [!h]                                                                                                                                                          
   \begin{center}                                                                                                                                                         
   \scalebox{0.9}{                                                                                                                                                       
   \begin{tikzpicture}[->,thick]                                                                                                                                          
    \node[initial,dia,initial text={$a=0$}] at (-4,-1) (A) 
    {$A5$} ;  
           \node[box,label=below:{$b=0$}] at (0,-1) (B) {B5};
       \node[cir,label=above:{$b=0$}] at (2,-1) (C) {J5};
       \node[cir,label=above:{$b=0$}] at (4,-1) (L) {K5};
       \node[cir] at (4,1) (O1) {};
       \node[cir,label=below:{$b\leq 1$}] at (6,-1) (M) {L5};
       \node[cir,accepting] at (8,-1) (N) {};
       \node[cir] at (6,1) (O) {};
       \path (C) edge node [midway,above] {}(L);
       \path (L) edge node [midway,above] {}(M);
       \path (L) edge node [midway,above] {}(O1);
       \path (M) edge node [midway,above] {$a \leq 1$}(N);
       \path (M) edge node [midway,right] {$a>1$}(O);
       
       \node[dia] at (2,-3) (P) {P5};
       \node[dia] at (4,-3) (Q) {Q5};
       \node[cir,label=above:{$b \leq 1$}] at (6,-3) (R) {R5};
       \node[cir,accepting] at (8,-3) (S) {};
       \node[cir] at (6,-5) (T) {};
       \path (C) edge node [midway,above] {}(P);
       \path (P) edge node [midway,above] {$a=1$} node [midway,below] {$a:=0$}(Q);
       \path (Q) edge node [midway,above] {$x_1=2$}  node [midway,below] {$b:=0$} (R);
       \path (R) edge node [midway,above] {$a \leq 1$}(S);
       \path (R) edge node [midway,left] {$a > 1$}(T);
       
        \path (A) edge node [midway,above] {$x_1\leq 1$}(B);
        \path (A) edge node [midway,below] {$b:=0$}(B);
        \path (B) edge[bend right=70] node [midway,above] {$x_1 \leq 1$}(A);
        \path (B) edge[bend right=70] node [midway,below] {$a:=0$}(A);
        \path (B) edge node [midway,below] {}(C);
    \node[dia] at (-7,-3) (F){} ; 
        
    \node[box,label=left:{$b=0$}] at (-4,-3) (E){$E5$} ; 
     \path (A) edge node [midway,left] {$x_1\leq 1$}(E);
     \path (A) edge node [midway,right] {$b:=0$}(E);
     \path (A) edge node [midway,left] {$x_1>1$}(F);
        
    \node[cir,label=above:{$b=0$}] at (-2,-3) (G){$G5$} ; 
    \node[cir,accepting] at (0,-3) (H){$H5$} ;
    \node[cir] at (-2,-5) (I){$I5$} ;
     \path (E) edge node [midway,above] {$x_1=1$}(G);
     \path (G) edge node [midway,left] {}(H);
     \path (G) edge node [midway,left] {}(I);
     
    \node[dia] at (-4,-5) (K){$R5$} ; 
     \path (E) edge node [midway,left] {$x_1<1$}(K);

        \end{tikzpicture}                                                                                                                                                      
  }                                                                                                                                                                     
  \caption{Wid$_{=0}$. $A5$ is the start node. $J5$ is entered with $a=t$, $x_1=\frac{1}{2^{c_1}3^{c_2}}+t$ and $b=0$, $t$ is the time spent at $A5$.
  Probability to reach a target location from $J5$ is $\frac{1}{4}(1-t)+\frac{1}{2}\frac{1}{2^{c_1}3^{c_2}}$ which is $\frac{1}{2}$ iff $t = 2 * \frac{1}{2^{c_1}3^{c_2}}$. This verifies that value $\frac{1}{2^{c_1}3^{c_2}}$ in $x_1$ is multiplied by 3 each time the loop is taken, 
  since $x_1$ becomes $\frac{1}{2^{c_1}3^{c_2}}+t=\frac{1}{2^{c_1}3^{c_2-1}}$. 
  Wid$_{>0}$ is obtained simply  having a multiply by 2 module at $R5$.
  }                                                                                                                             
   \label{wid-0}                                                                                                                                                        
  \end{center}                                                                                                                                                           
\end{figure}
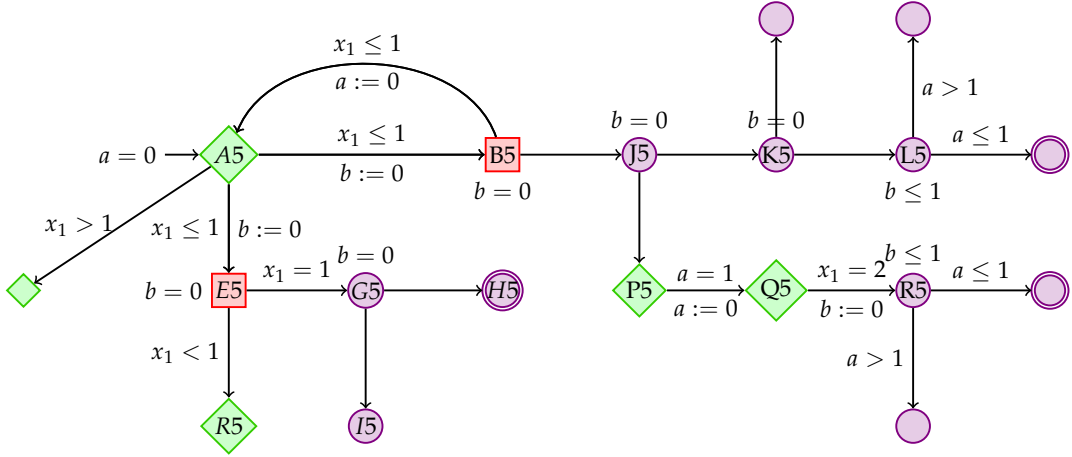

\paragraph*{Time Elapse for Zero Check}
Let us start looking at the main module for zero check in Figure \ref{zer-chk}.
 Assume that this is the $(k+1)$th instruction. A time $t_1+t_2=\frac{1}{2^{k+1}}$ is spent at locations 
 $\ell_{k+1}, B$ in Figure \ref{zer-chk}. Following this, if player $\Box$ goes in for a check in widgets 
 Check $z$ or Check $x$, the time elapse in these widgets is $<4$ as seen in the Increment section. 
 If not, control reaches one of the player $\Box$ locations $=0$ or $>0$. Here again, player $\Box$ can either 
 go ahead, or enter the Rem$_k^x$ widget.
 
 The Rem$_k^{x_1}$ widget is entered with $a=\frac{1}{2^{k+1}}$, $x_2=\frac{1}{2^{k+1+c_1} 3^{k+1+c_2}}$, 
 $z,x_1,b=0$. The time spent at $A0,B2$ is $\frac{1}{2^{k+1}}$. When control comes to the $\Box$ node, 
 there are two possibilities: (1) player $\Box$ continues with the Rem$_k^{x_2}$ widget, in which case 
 a time $\frac{1}{2^{k+1}}$ is elapsed. This can continue till $a=\frac{1}{2}$, a time 
 $\frac{1}{2^{k+1}}+\frac{1}{2^{k}}+\dots + \frac{1}{2^2}$ is elapsed after which, 
 the target $T$ is reached with probability $\frac{1}{2}$, or continues  
  till $a=1$ with time elapse $\frac{1}{2^{k+1}}+\frac{1}{2^{k}}+\dots + \frac{1}{2}$ and 
  control goes into the widget Wid$_{=0}$ or Mul $a$ or Mul $x_2$. The time elapse in the Mul $a$, 
  Mul $x_2$ widgets is $<4$. In the case of Wid$_{=0}$,  if the loop $B5-A5$ is taken till 
  $x_1=1$, a time $\frac{1}{2^{c_1}3^{c_2-1}}+\frac{1}{2^{c_1}3^{c_2-2}} + \dots \frac{1}{2^{c_1	}}$ is spent 
  till target $H5$  is reached. However, if player $\Box$ reaches out to the part from node $J5$, the time elapse is atmost 2 
  to reach a target. Thus, summing up, the time elapse is 
  \begin{enumerate}
\item Assume that the zero check is  the $(k+1)$th instruction. The time that has elapsed till the start of this instruction is 
$\frac{1}{2}+\dots +\frac{1}{2^{k}}$. 
\item The time elapse in the main module for zero check is  $\frac{1}{2^{k+1}}$. If player $\Box$ continues with the simulation, 
we are done.
\item If player $\Box$ enters any of the widgets (Check $x$, Check $z$, Rem$_k^{x_1}$, Rem$_k^{x_2}$, Wid$_{=0}$
Wid$_{>0}$), the time elapse is $<4$ till a target is reached.
\item The total time elapse till completion of $(k+1)$ instructions is thus
$<\frac{1}{2}+\dots +\frac{1}{2^{k}}+\frac{1}{2^{k+1}}+4$. 
\end{enumerate}

%% file: halt-final.tex
The gadget corresponding to the halt instruction is as follows:
 Once we reach the halt instruction, we go to a stochastic node $A$ with no time delay. 
 $A$ has two outgoing edges, one which leads to a target node, and the other one to a non-target.
 With no delay at $A$, the target is reached with probability $\frac{1}{2}$.  
We quickly give an intuition behind the proof of correctness of this construction:
Assume that the two counter machine halts. If Player $\Diamond$ simulates all the instructions correctly,  there are two possibilities:
 \begin{enumerate}
 \item Player $\Box$ allows simulation of the next instruction without
   entering any of the check gadgets. Then
   we will reach the halt location from where the probability to reach the target  is
   indeed $\frac{1}{2}$. \item Player $\Box$ enters any of the check
   gadgets during the simulation of some instruction. As can be seen
   from
   our earlier detailed
   analysis, 
   it is indeed the case that the probability to reach a target
   location is $\frac{1}{2}$.
   \end{enumerate} 
   
   Assume now that the two counter machine does not halt. If Player
   $\Diamond$ indeed simulates all the instructions correctly once
   again, then the only way to reach any target location is only by
   invoking a check gadget by Player $\Box$. As said above, clearly,
   this probability will be $\frac{1}{2}$ due to the correct
   simulation of Player $\Diamond$. Again, note that the times spent
   during increment/decrement of the $(k+1)$th instruction is
   $\frac{1}{2^{k+1}}$. This fact can be verified by the gadget
   $Check\_z$. In case of non-halting, therefore, the total time taken
   will converge to $1$. Thus, the time taken to reach any target
   location is $\leq 1$ in case of non-halting and correct simulation
   by Player $\Diamond$. Ofcourse, if Player $\Box$ never chooses to
   enter any of the check gadgets, then Player $\Diamond$ can never
   reach a target location, and hence cannot win.
 The total elapse in case Player $\Box$ enters a check gadget in the $(k+1)$th instruction is 
 $<4+\frac{1}{2}+\dots+\frac{1}{2^{k+1}}<5$.

 In both cases, if Player $\Diamond$ does not simulate correctly the
 instruction, Player $\Box$ can decide to check and the probablity to
 reach a target location will be $< \frac{1}{2}$.
Hence, Player $\Diamond$ has a winning strategy to ensure probability
$\frac{1}{2}$ for reaching a target location within $\Delta=5$ time
units iff the two-counter
machine halts.

%% file: main-arxiv.bbl
\begin{thebibliography}{10}

\bibitem{AD94}
R.~Alur and D.L. Dill.
\newblock A theory of timed automata.
\newblock {\em Theoretical Computer Science}, 126(2):183--235, 1994.

\bibitem{AMPS98}
E.~Asarin, O.~Maler, A.~Pnueli, and J.~Sifakis.
\newblock Controller synthesis for timed automata.
\newblock In {\em Proc. of IFAC Symposium on System Structure and Control},
  pages 469--474. Elsevier, 1998.

\bibitem{lics08}
C.~Baier, P.~Bouyer, T.~Brihaye, and M.~Gr{\"o}{\ss}er.
\newblock Almost-sure model checking of infinite paths in one-clock timed
  automata.
\newblock In {\em Proc. 23rd Annual Symposium on Logic in Computer Science
  (LICS'08)}, pages 217--226. IEEE Computer Society Press, 2008.

\bibitem{BHHK03}
C.~Baier, B.~Haverkort, H.~Hermanns, and J.-P. Katoen.
\newblock Model-checking algorithms for continuous-time {M}arkov chains.
\newblock {\em IEEE Transactions on Software Engineering}, 29(7):524--541,
  2003.

\bibitem{BK}
C~Baier and J.-P. Katoen.
\newblock {\em Principles of Model Checking}.
\newblock MIT Press, 2008.

\bibitem{BBBC16}
N.~Bertrand, P.~Bouyer, T.~Brihaye, and P.~Carlier.
\newblock Analysing decisive stochastic processes.
\newblock In {\em Proc. 43rd International Colloquium on Automata, Languages
  and Programming (ICALP'16)~-- {P}art~{II}}, Leibniz International Proceedings
  in Informatics. Leibniz-Zentrum f{\"u}r Informatik, July 2016.
\newblock To appear.

\bibitem{qest08}
N.~Bertrand, P.~Bouyer, T.~Brihaye, and N.~Markey.
\newblock Quantitative model-checking of one-clock timed automata under
  probabilistic semantics.
\newblock In {\em Proc. 5th International Conference on Quantitative Evaluation
  of Systems (QEST'08)}. IEEE Computer Society Press, 2008.

\bibitem{journal-sta}
N.~Bertrand, P.~Bouyer, T.~Brihaye, Q.~Menet, M.~Gr{\"o}{\ss}er, and
  M.~Jurdzi{\'n}ski.
\newblock Stochastic timed automata.
\newblock {\em Logical Methods in Computer Science}, 10(4):1--73, 2014.

\bibitem{BBG14}
N.~Bertrand, T.~Brihaye, and B.~Genest.
\newblock Deciding the value 1 problem for reachability in 1-clock decision
  stochastic timed automata.
\newblock In {\em Proc. 11th International Conference on Quantitative
  Evaluation of Systems (QEST'14)}, pages 313--328. IEEE Computer Society
  Press, 2014.

\bibitem{BDHK06}
H.C. Bohnenkamp, P.R. D'Argenio, H.~Hermanns, and J.-P. Katoen.
\newblock {MODEST:} {A} compositional modeling formalism for hard and softly
  timed systems.
\newblock {\em IEEE Transactions on Software Engineering}, 32(10):812--830,
  2006.

\bibitem{icalp09}
P.~Bouyer and V.~Forejt.
\newblock Reachability in stochastic timed games.
\newblock In {\em Proc. 36th International Colloquium on Automata, Languages
  and Programming (ICALP'09)}, volume 5556 of {\em LNCS}, pages 103--114.
  Springer, 2009.

\bibitem{BKK+11b}
Tom{\'a}\v{s} Br{\'a}zdil, Jan Kr\v{c}{\'a}l, Jan K\v{r}et{\'i}nsk{\'y}, and
  Vojt\v{e}ch {\v{R}}eh{\'a}k.
\newblock Fixed-delay events in generalized semi-{M}arkov processes revisited.
\newblock In {\em Proc. 22nd International Conference on Concurrency Theory
  (CONCUR'11)}, volume 6901 of {\em LNCS}, pages 140--155. Springer, 2011.

\bibitem{concur14}
T.~Brihaye, G.~Geeraerts, S.~N. Krishna, L.~Manasa, B.~Monmege, and A.~Trivedi.
\newblock Adding negative prices to priced timed games.
\newblock In {\em Proc. 25th International Conference on Concurrency Theory
  (CONCUR'14)}, LIPIcs, pages 560--575. Leibniz-Zentrum f{\"u}r Informatik,
  2014.

\bibitem{FV97}
J.~Filar and K.~Vrieze.
\newblock {\em Competitive {M}arkov Decision Processes}.
\newblock Springer, 1997.

\bibitem{FKNT10}
V.~Forejt, M.~Kwiatkowska, G.~Norman, and A.~Trivedi.
\newblock Expected reachability-time games.
\newblock In {\em Proc. 8th International Conference on Formal Modeling and
  Analysis of Timed Systems (FORMATS'10)}, volume 6246 of {\em LNCS}, pages
  122--136. Springer, 2010.

\bibitem{KNSS00}
M.~Kwiatkowska, G.~Norman, R.~Segala, and J.~Sproston.
\newblock Verifying quantitative properties of continuous probabilistic timed
  automata.
\newblock In {\em Proc. of 11th International Conference on Concurrency
  Theorey, (CONCUR'00)}, volume 1877 of {\em LNCS}, pages 123--137. Springer,
  2000.

\bibitem{KNSS02}
M.~Kwiatkowska, G.~Norman, R.~Segala, and J.~Sproston.
\newblock Automatic verification of real-time systems with discrete probability
  distributions.
\newblock {\em Theoretical Computer Science}, 282(1):101--150, June 2002.

\bibitem{LMS04}
F.~Laroussinie, N.~Markey, and P.~Schnoebelen.
\newblock Model checking timed automata with one or two clocks.
\newblock In {\em Proc. 15th International Conference on Concurrency Theory
  (CONCUR'04)}, volume 3170 of {\em LNCS}, pages 387--401. Springer, 2004.

\bibitem{Min67}
M.~Minsky.
\newblock {\em Computation: Finite and Infinite Machines}.
\newblock Prentice Hall International, 1967.

\bibitem{ORW09}
J.~Ouaknine, A.~Rabinovich, and J.~Worrell.
\newblock Time-bounded verification.
\newblock In {\em Proc. 20th International Conference on Concurrency Theory
  (CONCUR'09)}, volume 5710 of {\em LNCS}, pages 496--510. Springer, 2009.

\bibitem{OW10}
J.~Ouaknine and J.~Worrell.
\newblock Towards a theory of time-bounded verification.
\newblock In {\em Proc. 37th International Colloquium on Automata, Languages
  and Programming (ICALP'10)}, volume 6199 of {\em LNCS}, pages 22--37.
  Springer, 2010.

\bibitem{Put94}
M.~L. Puterman.
\newblock {\em {M}arkov Decision Processes: Discrete Stochastic Dynamic
  Programming}.
\newblock Wiley, 1994.

\bibitem{UP01}
Uppaal case-studies.
\newblock
  {\footnotesize{\url{http://www.it.uu.se/research/group/darts/uppaal/examples.shtml}}}.

\end{thebibliography}
